\documentclass[9pt]{amsart}
\usepackage[english]{babel}
\usepackage[latin1]{inputenc}
\usepackage{mathrsfs} %<--pacchetto per lettere in corsivo
\usepackage{amsmath}
\usepackage{bbm}
\usepackage{amssymb} %<--pacchetto pi\~{A}\^{A}¹ ampio per simboli matematici

\usepackage{graphicx}
\usepackage{subcaption}
\usepackage[rightcaption]{sidecap}

\usepackage{titlesec} 

\titleformat{\section}[runin]
  {\normalfont\large\bfseries}{\thesection}{.5em}{}
\titleformat{\subsection}[runin]
  {\normalfont\large\bfseries}{\thesubsection}{.5em}{}

\usepackage{amsthm}% http://ctan.org/pkg/amsthm
\theoremstyle{theorem}\newtheorem{theorem}{Theorem}
\theoremstyle{definition}\newtheorem{definition}{Definition}
\theoremstyle{definition}
\theoremstyle{theorem}\newtheorem{lemma}{Lemma}
\theoremstyle{theorem}\newtheorem{proposition}{Proposition}
\theoremstyle{theorem}
\theoremstyle{remark}\newtheorem{remark}{Remark}
\theoremstyle{remark}\newtheorem{example}{Example}

\usepackage{lipsum} % Used for inserting dummy 'Lorem ipsum' text into the template

\usepackage[T1]{fontenc}

\title{Optimal Bond Portfolio under constraints}

\title{Optimal Model Points choice in Life Insurance}

\title{Optimal characterization of bond portfolios and model points in Life Insurance}

\title{Optimal model points portfolio in Life Insurance}

\title{Infinite dimensional portfolio representation as applied to model points selection in life insurance}

\author{Enrico Ferri}
\address{Banco Santander, Madrid}
\email{enrico.ferri@gruposantander.com}

\date{\today}
\makeindex
\thanks{Written during author's stay at Universidade da Coru\~na. This work has been funded by EU H2020-ITN-EID-2014 WAKEUPCALL (Grant Agreement 643045).}

\begin{document}

\maketitle

\begin{abstract}
We consider the problem of seeking an optimal set of model points associated to a fixed portfolio of life insurance policies. Such an optimal set is characterized by minimizing a certain risk functional, which gauges the average discrepancy with the fixed portfolio in terms of the fluctuation of the interest rate term structure within a given time horizon. We prove a representation theorem which provides two alternative formulations of the risk functional and which may be understood in connection with the standard approaches for the portfolio immunization based on sensitivity analysis. For this purpose, a general framework concerning some techniques of stochastic integration in Banach space and Malliavin calculus is introduced. A numerical example is discussed when considering a portfolio of whole life policies. 
\end{abstract}

\

\section{Introduction.} 
This paper is motivated by problem of the efficient 
portfolio representation, in which one seeks to substitute a portfolio of market securities with a simpler one which owns similar risk, when certain contingent restrictions are permitted. 
%characterizing a suitable portfolio of market securities with the same performance as a given financial exposure, 
One encounters this question when defining a hedging strategy subject to policy and budget constraints or reducing the scale, and hence the complexity, of a specific portfolio for analysis and management purposes, without misrepresenting its inherent risk structure. 
	
As the main purpose of this article, we set up the problem of replacing a given portfolio of life insurance policies by considering a small group of representative contracts, usually known as the related model points. Life insurance companies are allowed by regulators to estimate the performance of any portfolio of policies on the basis of suitable model points in order to reduce the computational difficulties of the operation, provided that it does not result in the loss of any significant attribute of the portfolio itself.

We address this problem by defining a reasonable notion of optimality based on a portfolio comparison criterion. 
In particular, we present an approach consisting in the minimization of a certain risk functional, which gauges the average discrepancy between the original portfolio and a given set of model points, in terms of the fluctuation of the interest rate term structure within a given time horizon.  

%This idea is meaningful, since the value of a life insurance portfolio is required to be estimated as the weighed average of the discounted future cash flows. 

We show two different formulations of this functional within the theory of the stochastic integration in UMD Banach spaces advanced in \cite{vanweis}, and by considering some tools of Malliavin calculus developed in this framework, as presented in \cite{maas2,maas} and \cite{pronk}. 
The main idea is to follow the approach considered in \cite{carteh,car} and \cite{eke}, by modelling the discounted price curve as an infinite-dimensional dynamics in a Banach space of continuous functions driven by a cylindrical Wiener process. Thus, any bond portfolio may be represented as an element of the dual of such a space. Within this framework, a natural tool to characterize the diffusion component of the dynamics is via the Banach space of $\gamma$-radonifying operators, which provides an extension of the Hilbert-Schmidt operators class considered in \cite{carteh,car}.

The first formulation is shown in terms of the Malliavin derivative operator.  
In particular, the criterion we obtain turns out to be similar to the minimization approach suggested in \cite{ket} to address the problem of the optimal hedging of bond portfolios, in which a refined notion of duration is introduced by using Malliavin calculus for Gaussian random fields in the Hilbert space framework. 
The second formulation is obtained under further conditions on the model and it mainly involves the diffusive component of the discount curve dynamics. 

Both alternative formulations we present are assessed by invoking notions of differential calculus in Banach spaces. In this respect, they naturally generalize the standard techniques for portfolio immunization based on sensitivity analysis.

\

This paper is organized in the following way.
Section \ref{sec_preliminaries} reviews the results within theory of stochastic integration in UMD Banach spaces and the notions of Malliavin calculus we will use throughout the paper. Although most of the results we present in this section are known, for convenience, we prove those ones we could not find in a suitable form in the existing literature.
Section \ref{sec:optimlaity} collects the main mathematical results of the paper. In particular, Theorem \ref{pr:optimal_strategy} which shows the equivalence of the two different formulations discussed above. 
Sections \ref{sec:financial_setup} and \ref{sec:examples} are dedicated to a detailed description of the setup we propose from the financial point of view. Further, a couple of examples regarding standard problems in portfolio hedging are discussed throughout this approach.  
Section \ref{sec: Interest Rates} shows a direct application when dealing with the theory of the optimal hedging portfolio within the fixed income framework. 
In particular, we discuss a numerical example in which a refined interpretation of the bond duration naturally arises.
Section \ref{sec:life_insurance} is entirely devoted to the problem of characterizing an optimal set of model points when considering a generic portfolio of homogeneous life insurance policies. Then, a similar numerical example is discussed when considering a portfolio of whole life policies.
  
\

\section{Setting.} \label{sec_preliminaries}
All the vector spaces we consider are assumed to be real. Given a vector space $S$ and real values $a$ and $b$, we write $a\lesssim_S b$ to indicate that there exists some constant $\beta$ only depending on $S$ such that $a\leq \beta b$. Further, we write $a\triangleq b$, when the identity $a=b$ holds by definition.

\ 
 
\noindent{\emph{Notation}.} 
Here and in the sequel of the article, we fix a separable Hilbert space $H$ and we write $\langle\cdot,\cdot\rangle_H$ to denote its inner product. We will always identify $H$ with its dual via the Riesz representation theorem. Further, we consider a Banach space $E$ together with its dual $E^\ast$. The duality pairing between $E$ and $E^\ast$ is denoted by $\langle \cdot, \cdot \rangle_E$. Moreover, we write $\mathcal{L}(H,E)$ to denote the space of bounded and linear operators mapping $H$ into $E$.

Let $(\Omega,\mathscr{F},\mathbb{P})$ be a reference complete probability space and write $I$ to denote the unit interval on the real line. 
%We denote by $L^0(\Omega;E)$ the vector space of the equivalence classes of $E$-valued random variables, such that two variables are in the same class if they agree a.s. on $\Omega$. 
Throughout this article, an $E$-valued process is a one-parameter family of $E$-valued random variables indexed by  $I$. In most cases, we identify the generic $E$-valued process with the induced map $\Omega\times I\rightarrow E$. Further, a $\mathcal{L}(H,E)$-valued process $\theta\triangleq \lbrace \theta_t : t\in I \rbrace$ is said $H$-strongly measurable if the $E$-valued process $\theta h\triangleq \lbrace \theta_t h, t\in I\rbrace$ is strongly measurable, for any $h\in H$.

We fix a cylindrical $H$-Wiener process $W\triangleq\lbrace W_t : t \in I \rbrace$. i.e. a one-parameter family of bounded and linear operators from $H$ to $L^2(\Omega)$, satisfying the usual conditions:
\begin{itemize}
\item[(i)] For any $h\in H$, the process $Wh\triangleq \lbrace W_th: t\in I \rbrace$ is a standard Brownian motion;
\item[(ii)] For any $t,s\in I$ and $h_1,h_2\in H$, one has that $\mathbb{E}\lbrace W_th_1W_sh_2\rbrace=(s\wedge t)\langle h_1,h_2\rangle_H$.
\end{itemize}

Thus, we define $\mathscr{G}^W\triangleq\lbrace\mathscr{G}^W_t : t \in I \rbrace$ to be the augmented filtration generated by the $H$-Wiener process $W$. We say that an $E$-valued process is adapted when it is adapted to the filtration $\mathscr{G}^W$. On the other hand, an $H$-strongly measurable process $\theta:\Omega\times I \mapsto \mathcal{L}(H,E)$ is said to be adapted if the $E$-valued process $\theta h$ is adapted, for any $h\in H$. 

\
 
\noindent{\emph{$\gamma$-Radonifying operators.}}
We write $\gamma(H,E)$ to denote the subspace of $\gamma$-radonifying operators in $\mathcal{L}(H,E)$. In particular, one has that $\vartheta \in \gamma(H,E)$ if for some (or, equivalently, for any) orthonormal basis $h_1,h_2,...$ of $H$, the random sum $\sum_n \gamma_n\vartheta h_n$ converges in the topology of $L^2(\Omega;E)$, where $\gamma_1,\gamma_2,...$ represents a sequence of independent and real-valued random variables, such that the law of $\gamma_n$ is standard Gaussian, for any $n\geq 1$. 

The space $\gamma(H,E)$ is a Banach space, when endowed with the norm defined by,
\begin{equation}
\Vert \vartheta \Vert_{\gamma(H,E)} \triangleq \bigg \lbrace \mathbb{E}\bigg\Vert \sum_{n\geq 1} \gamma_n \vartheta h_n\bigg\Vert_E^2 \bigg\rbrace^{1/2}, \hspace{1cm} \text{for any $\vartheta\in \gamma(H,E)$} \nonumber.
\end{equation} 

We recall that $\gamma(H,E)$ is an operator ideal in $\mathcal{L}(H,E)$, i.e. the following ideal property holds true.
\begin{lemma} \label{le:ideal_property}
Let $H_1$ and $H_2$ be Hilbert spaces and $E_1$ and $E_2$ Banach spaces. Let $S\in\mathcal{L}(H_2,H_1)$ and $T\in \mathcal{L}(E_1,E_2)$. Then, if $\vartheta\in\gamma(H_1,E_1)$ one has $T \vartheta S\in \gamma(H_2,E_2)$ and moreover the following norm equality holds true, 
\begin{equation}
\Vert T\vartheta S \Vert_{\gamma(H_2,E_2)}\leq \Vert T \Vert_{ \mathcal{L}(E_1,E_2)} \Vert \vartheta \Vert_{\gamma(H_1,E_1)} \Vert S \Vert_{\mathcal{L}(H_2,H_1)}. \nonumber
\end{equation}
\begin{proof}
See, e.g., Theorem 6.2 in \cite{vanner_gamma}.
\end{proof}
\end{lemma}
In the case when $\vartheta\in \gamma(H,E)$ and $x^\ast\in E^\ast$, we write $\langle x^\ast,\vartheta \rangle_E$ to denote the $H$-valued dual pairing between $\gamma(H,E)$ and $E^\ast$, defined by setting $\langle x^\ast,\vartheta \rangle_E\triangleq \vartheta^\ast x^\ast$, where $\vartheta^\ast$ denotes the Banach space adjoint operator of $\vartheta$.
On the other hand, from Lemma \ref{le:ideal_property}, one has that
\begin{equation}\label{eq:ideal_dual}
\Vert \langle x^\ast,\vartheta\rangle_E \Vert_H \leq \Vert x^\ast\Vert_{E^\ast} \Vert \vartheta \Vert_{\gamma(H,E)}.
\end{equation}

We recall that, if $E$ is further assumed to be a Hilbert space, then $\gamma(H,E)$ boils down to the space of Hilbert-Schmidt operators mapping $H$ into $E$. Hence, we have the natural identifications $\gamma(H,\mathbb{R})=H$ and $\gamma(\mathbb{R},E)=E$, endowed with the related norms. 

The following lemma introduces the notation for the trace operator $\text{tr}(\cdot;\cdot)$. 
\begin{lemma}
Let $h_1,h_2,...$ be an orthonormal basis of $H$.
For any $T\in \mathcal{L}(E,E^\ast)$ and $\vartheta\in \gamma(H,E)$, the series  
\begin{equation}\label{eq:trace_operator}
\text{tr}(T;\vartheta)\triangleq \sum_{n\geq 1} \langle T(\vartheta h_n),\vartheta h_n\rangle_E,
\end{equation}
converges and its sum does not depend on the choice of the orthonormal basis $h_1,h_2,...$ of $H$.
\begin{proof}
See, e.g., Lemma 2.3 in \cite{brz}.
\end{proof}
\end{lemma}

We refer to \cite{vanner_gamma} and \cite{vanner} for an exhaustive description of the space $\gamma(H,E)$ and further properties. 

\ 

\noindent{\emph{Mallivian Derivative.}}
Throughout this paper, we write $\mathscr{H}=L^2(I;H)$. 

For any $E$-valued variable $Z$ differentiable in the Malliavin sense, we write $D Z$ for its Malliavin derivative. More precisely, we understand the Malliavin derivative as a closable operator from $L^2(\Omega;E)$ into $L^2(\Omega;\gamma(\mathscr{H},E))$, (see, e.g., \cite{maas2,maas} and \cite{pronk}). We denote by $D$ such a closure and by $\mathbb{H}^{1,2}(E)$ the domain of $D$ in $L^2(\Omega;E)$. 

We recall that $\mathbb{H}^{1,2}(E)$ is a Banach space when it is endowed with the norm 
\begin{equation}
\Vert Z \Vert_{\mathbb{H}^{1,2}(E)} \triangleq \lbrace\Vert Z\Vert_{L^{2}(\Omega; E)}^2 + \Vert D Z \Vert_{L^{2}(\Omega; \gamma(\mathscr{H},E))}^2 \rbrace^{1/2}, \hspace{1cm} \text{for any $Z\in \mathbb{H}^{1,2}(E)$.} \nonumber
\end{equation} 

Similarly, $\mathbb{H}^{2,2}(E)$ denotes the domain of $D^2\triangleq D\circ D$ in $L^2(\Omega,E)$. In the particular case when $E=\mathbb{R}$, we write $\mathbb{H}^{k,2}\triangleq\mathbb{H}^{k,2}(\mathbb{R})$, for $k=1,2$.\\

\noindent{\emph{Stochastic evolution.}}
Throughout this paper, we always assume $E$ to be a UMD space with type 2 (see \cite{maas} and \cite{vanner}). 

Let $\xi_0\in L^2(\Omega;E)$ be a strongly $\mathscr{G}^W_0$-measurable random variable. Consider an adapted and strongly measurable $E$-valued stochastic process $b=\lbrace b_t : t\in I \rbrace$, that belongs to $L^2(\Omega;L^2(I;E))$. Let $\sigma=\lbrace \sigma_t : t\in I \rbrace$ be some adapted and $H$-strongly measurable $\mathcal{L}(H,E)$-valued process that belongs to $L^2(\Omega;L^2(I;\gamma(H,E)))$. 

Moreover, we suppose that the following conditions hold true,
\begin{equation}\label{eq:malliavin_hypotesis}
\xi_0\in\mathbb{H}^{1,2}(E), \ \ \ b\in \mathbb{H}^{1,2}(L^2(I;E)), \ \ \ \sigma \in \mathbb{H}^{2,2}(L^2(I;\gamma(H,E))).
\end{equation}

The result below will be useful later on in this paper.
\begin{lemma}\label{le:embedding_gamma}
The process $\sigma$ is well defined as an element of $L^2(\Omega;\gamma(\mathscr{H},E))$. Moreover, for $k=1,2$, we have that $\sigma\in \mathbb{H}^{k,2}(\gamma(\mathscr{H},E))$, and the following norm inequality holds,
\begin{equation}\label{eq:embedding_gamma_norm_inequality}
\Vert \sigma \Vert_{\mathbb{H}^{k,2}(\gamma(\mathscr{H},E))}\lesssim_{E} \Vert \sigma \Vert_{\mathbb{H}^{k,2}(L^2(I,\gamma(H,E)))}. 
\end{equation}
\begin{proof}
Since the space $E$ is assumed to have type 2, there exists a continuous and linear embedding $$\imath_2 : L^2(I,\gamma(H,E)) \hookrightarrow \gamma(\mathscr{H},E),$$ 
with operatorial norm satisfying $\Vert \imath_2 \Vert \leq \beta_2$, where $\beta_2$ denotes the type $2$ constant of $E$, (see \cite{vanweis2}, Lemma 6.1). Thus, the process $\sigma$ turns out to be well defined as an element of $L^2(\Omega;\gamma(\mathscr{H},E))$. Moreover, one has $\Vert \cdot \Vert_{\gamma(\mathscr{H},E)} \lesssim_E  \Vert \cdot \Vert_{L^2(I;\gamma(H,E))}$, and hence for $k=1,2$, the following inequality holds true,
\begin{equation}\label{eq:embedding_gamma:norm:inequality2}
\Vert \cdot \Vert_{\mathbb{H}^{k,2}(\gamma(\mathscr{H},E))}\lesssim_E \Vert \cdot \Vert_{\mathbb{H}^{k,2}(L^2(I,\gamma(H,E)))}. 
\end{equation}

Thus, since we assumed that $\sigma \in \mathbb{H}^{2,2}(L^2(I;\gamma(H,E)))$, from inequality (\ref{eq:embedding_gamma:norm:inequality2}) we obtain that $\sigma \in \mathbb{H}^{1,2}(\gamma(\mathscr{H},E))$.
\end{proof}
\end{lemma}

We will use the notion of stochastic integration as defined in 
\cite{vanner} and \cite{vanweis}.
Moreover, for any $\mathcal{L}(H,E)$-valued stochastically integrable process $\psi\triangleq \lbrace \psi_t : t\in I \rbrace$, we will write
\begin{equation}
\int_I \psi_t dW_t = \delta (\psi), \nonumber
\end{equation}
where $\delta$ denotes the divergence operator defined on $L^2(\Omega;\gamma(\mathscr{H},E))$, (see \cite{maas}, Theorem 5.4). 

For any stochastically integrable $\mathcal{L}(H,\mathbb{R})$-valued process, the version of the It\^o's isometry is regarded as follows.
\begin{lemma}\label{le:ito_isometry}
Let $\psi\triangleq \lbrace \psi_t : t\in I \rbrace$ be an adapted $H$-strongly measurable and stochastically integrable process, taking values in $\mathcal{L}(H,\mathbb{R})$. Then, for any $t\in I$, one has
\begin{equation}
\mathbb{E} \bigg\lbrace \bigg\vert\int_0^t \psi_s dW_s \bigg\vert^2\bigg\rbrace=\mathbb{E}\bigg\lbrace \int_0^t \Vert\psi_s\Vert^2_H ds \bigg\rbrace. \nonumber
\end{equation}
\begin{proof}
Fix $t\geq 1$ and let $h_n$, for $n\geq 1$, be an orthonormal basis of $H$. Then, for any $n\geq 1$, the process $\psi h_n\triangleq \lbrace \langle\psi_t,h_n\rangle_H : t\in I \rbrace$ is stochastically integrable with respect to $Wh_n\triangleq \lbrace W_th_n : t\in I \rbrace$, and for any $t\in I$ the following representation holds, 
\begin{equation}\label{eq:series_expansion_Ito_integral}
\int_0^t \psi_s dW_s = \sum_{n\geq 1}\int_0^t \langle\psi_s,h_n\rangle_H dW_sh_n, 
\end{equation} 
where the convergence of the series in (\ref{eq:series_expansion_Ito_integral}) is understood in the topology of $L^2(\Omega)$, (see \cite{vanner}, Corollary 3.9). 

Besides, since the processes $Wh_n$ and $Wh_m$ are independent for any $n,m\geq 1$ such that $n\neq m$, jointly with the It\^o's isometry, for any $t\in I$ we have that, 
\begin{eqnarray}\label{eq:ito_delta}
\mathbb{E}\bigg\lbrace \int_0^t \langle\psi_s,h_n\rangle_H dW_sh_n \cdot \int_0^t \langle\psi_s, h_m\rangle_H dW_s h_m \bigg\rbrace &=& \delta_{nm}  \mathbb{E} \bigg\lbrace \bigg\vert \int_0^t \langle\psi_s,h_n\rangle_H dW_sh_n \bigg\vert^2\bigg\rbrace   \nonumber \\
&=& \delta_{nm} \mathbb{E} \bigg \lbrace \int_0^t \vert \langle\psi_s,h_n\rangle_H \vert^2 ds \bigg\rbrace 
\end{eqnarray}
where $\delta_{nm}$ denotes the Kronecker delta $\delta_{nm}=1$ if $n=m$ and $\delta_{nm}=0$ otherwise. 

Thus, we have 
\begin{eqnarray}
\mathbb{E}\bigg\lbrace \bigg\vert \int_0^t \psi_s dW_s \bigg\vert^2\bigg\rbrace &\stackrel{\text{(i)}}{=}& \mathbb{E} \bigg\lbrace \bigg\vert \sum_{n\geq 1}\int_0^t \langle\psi_s,h_n\rangle_H dW_sh_n \bigg\vert^2\bigg\rbrace \nonumber\\
&=& \sum_{n\geq 1} \sum_{m\geq 1}\mathbb{E}\bigg\lbrace \int_0^t \langle\psi_s,h_n\rangle_H dW_sh_n \cdot \int_0^t \langle\psi_s, h_m\rangle_H dW_s h_m \bigg\rbrace \nonumber  \\
&\stackrel{\text{(ii)}}{=}& \sum_{n\geq 1} \mathbb{E} \bigg \lbrace \int_0^t \vert \langle \psi_s,h_n\rangle_H \vert^2 ds \bigg\rbrace \nonumber \\
&=& \mathbb{E} \bigg \lbrace \int_0^t \Vert \psi_s \Vert^2_H ds \bigg\rbrace. \nonumber
\end{eqnarray}
where in $\text{(i)}$ we have used the representation (\ref{eq:series_expansion_Ito_integral}) and in $\text{(ii)}$ the identity (\ref{eq:ito_delta}).
\end{proof}
\end{lemma}
\begin{lemma}
The process $\sigma$ is stochastically integrable. 
\begin{proof}
The result follows directly from Corollary 3.10 in \cite{vanner}, since  $\sigma\in L^2(\Omega;L^2(I;\gamma(H,E)))$.
%First off all, notice that the inequality (\ref{eq:ideal_dual}) assures that a.s.
%\begin{equation}
%\Vert \langle x^\ast,\sigma \rangle_E \Vert_{\mathscr{H}}\leq \Vert x^\ast \Vert_{E^\ast} \Vert \sigma\Vert_{\gamma(\mathscr{H},E)}, \hspace{1cm} \text{for any $x^\ast\in E^\ast$.} \nonumber
%\end{equation}
%
%Then, since $\sigma\in L^2(\Omega;\gamma(\mathscr{H},E))$ by Lemma \ref{le:embedding_gamma}, one has that  $\sigma\in L^2(\Omega; \mathscr{H})$ scalarly, i.e. $\langle x^\ast , \sigma \rangle_E \in L^2(I;\mathscr{H})$, for any $x^\ast\in E^\ast$.
%
%Hence, by invoking Theorem 3.6 in \cite{vanner}, we obtain that the process $\sigma$ is stochastically integrable.
\end{proof}
\end{lemma}

We consider the $E$-valued stochastic process $\xi\triangleq \lbrace \xi_t : t\in I \rbrace$ defined by setting,
\begin{equation} \label{eq:diffusive_dynamics_general}
\xi_t \triangleq \xi_0 + \int_0^t b_sds + \int_0^t\sigma_s dW_s, \hspace{1cm} \text{for any $t\in I$}. 
\end{equation}
The following result is taken from \cite{pronk}; we include a proof for convenience.

\begin{lemma}\label{le:xi_malliavin_domain}
For any $t\in I$, we have that $\xi_t\in \mathbb{H}^{1,2}(E)$.
\end{lemma}
First, we prove the following Lemma.
\begin{lemma}\label{le:xi_well_defined}
For any $t\in I$, we have that $\xi_t \in L^2(\Omega;E)$ is well defined, and moreover 
\begin{equation}
\sup_{t\in I}\Vert \xi_t \Vert^2_{L^2(\Omega;E)}< \infty. \nonumber
\end{equation}
\end{lemma}
\begin{proof}[Proof of Lemma \ref{le:xi_well_defined}]
Note that, since the operator $\delta$ is linear and continuous from $\mathbb{H}^{1,2}(\gamma(\mathscr{H},E))$ to $L^2(\Omega;E)$, (see \cite{pronk}, Proposition 4.3), we have by invoking Lemma \ref{le:embedding_gamma} that 
\begin{equation}
\Vert \delta(\sigma)\Vert_{L^2(\Omega;E)}\lesssim_E \Vert \sigma \Vert_{\mathbb{H}^{1,2}(\gamma(\mathscr{H},E))}\lesssim_E \Vert \sigma \Vert_{\mathbb{H}^{1,2}(L^2(I;\gamma(H,E)))}. \nonumber
\end{equation}
As a consequence, 
\begin{eqnarray}
\sup_{t\in I} \Vert \xi_t \Vert^2_{L^2(\Omega;E)} & \leq & \Vert \xi_0 \Vert^2_{L^2(\Omega;E)} + \Vert b \Vert^2_{L^2(\Omega ;L^2(I;E))} + \Vert \delta(\sigma) \Vert^2_{L^2(\Omega;E)}  \nonumber \\
&\lesssim_E & \Vert \xi_0 \Vert^2_{L^2(\Omega;E)} + \Vert b \Vert^2_{L^2(\Omega ;L^2(I;E))} + \Vert \sigma \Vert^2_{\mathbb{H}^{1,2}(L^2(I;\gamma(H,E)))},  \nonumber
\end{eqnarray}
and $\xi_t\in L^2(\Omega;E)$ is well defined, for any $t\in I$.
\end{proof}
As in \cite{pronk}, for any $t\in I$, we understand $\mathbbm{1}_{[0,t]}:\mathscr{H}\rightarrow \mathscr{H}$ as a  bounded and linear operator defined as
\begin{equation}
h\in \mathscr{H}\mapsto(\mathbbm{1}_{[0,t]} h)(\cdot) \triangleq \mathbbm{1}_{[0,t]}(\cdot)h(\cdot). \nonumber
\end{equation}
\begin{remark}\label{re:indicator_operator}
According to Lemma \ref{le:ideal_property}, we may regard $\mathbbm{1}_{[0,t]}$ as a well defined operator on $\gamma(\mathscr{H},\mathbb{H}^{2,2}(E))$, by setting $(\mathbbm{1}_B\vartheta)h\triangleq \vartheta(\mathbbm{1}_B h)$, for any $\vartheta\in \gamma(\mathscr{H},\mathbb{H}^{1,2}(E))$.
\end{remark}
\begin{proof}[Proof of Lemma \ref{le:xi_malliavin_domain}]

Note that by linearity it is enough to prove that $\delta(\mathbbm{1}_{[0,t]}\sigma)\in \mathbb{H}^{1,2}(E)$, for any $t\in I$, since $\xi_0\in \mathbb{H}^{1,2}(E)$ and $b\in \mathbb{H}^{1,2}(L^2(I,E))$.

First, we may regard $\sigma$ as an element of $\gamma(\mathscr{H},\mathbb{H}^{2,2}(E))$, since from Lemma \ref{le:embedding_gamma} we have $\sigma\in \mathbb{H}^{2,2}(\gamma(\mathscr{H};E))$ and the space $\mathbb{H}^{2,2}(\gamma(\mathscr{H},E))$ is isometric to $\gamma(\mathscr{H},\mathbb{H}^{2,2}(E))$, (see \cite{pronk}, Theorem 2.9).
Thus, according to Remark \ref{re:indicator_operator} we have that $\mathbbm{1}_{[0,t]}\sigma\in \mathbb{H}^{2,2}(\gamma(\mathscr{H},E)))$, for any $t\in I$, and hence that $\delta(\mathbbm{1}_{[0,t]}\sigma)\in \mathbb{H}^{1,2}(E)$, by Proposition 4.4 in \cite{pronk}.
\end{proof}

\

\section{Risk functional and optimization.} \label{sec:optimlaity}
Let $D$ be some UMD Banach space. A function $\psi:I\times E\rightarrow D$ is said to be of class $\mathscr{C}^{1,2}$ if it is differentiable in the first variable and twice continuously Fr\'echet differentiable in the second variable and the functions $\psi$, $\nabla_k \psi$, for $k=1,2$, and $\nabla_2^2 \psi$ are continuous on $I\times E$. 
Moreover, we shall say that $\psi$ is of class $\mathscr{C}_b^{1,2}$ when in addition the following condition is met,
\begin{equation}\label{eq:f_bound_contion}
 \Vert \nabla_2 \psi\Vert_\infty\triangleq\sup_{(t,x)\in I\times E}\Vert \nabla_2 \psi
 (t,x)\Vert_{\mathcal{L}(E,D)} < \infty.
\end{equation}
Here and in the sequel, we write $\nabla_k \psi$ to denote the derivative of $\psi$ with respect to the $k$th component, for any $k=1,2$.

Throughout this section, we always suppose that an $E$-valued process $\xi\triangleq \lbrace \xi_t : t\in I \rbrace$ as defined by means of the identity (\ref{eq:diffusive_dynamics_general}) is \emph{a priori} fixed.

\begin{definition}\label{de:BS-function}
We say that a function $\psi: I\times E\rightarrow D$ of class $\mathscr{C}^{1,2}$  is a \emph{BS-function} relative to $\xi$, if the following condition holds true a.s.
\begin{equation} \label{eq:heat_equationo_zeta}
 \nabla_1 \psi(t,\xi_t) +\frac{1}{2}\text{tr} (\nabla_2^2 \psi(t,\xi_t);\sigma_t)=0, \hspace{0.5cm}\text{for any $t\in I$.}
\end{equation}
\end{definition}
The result below characterizes the dynamics of the process $\psi(t,\xi_t)$, for $t\in I$, when $\psi$ is a BS-function relative to $\xi$.
\begin{lemma} \label{le:BL-function}
Let $\psi:I\times E\rightarrow D$ be a function of class $\mathscr{C}^{1,2}$. When $\psi$ is assumed to be a BS-function relative to $\xi$, one has 
\begin{equation} \label{eq:BL-condition}
\psi(t,\xi_t) = \psi(0,\xi_0) + \int_0^t \nabla_2 \psi(s,\xi_s) b_s ds + \int_0^t \nabla_2 \psi(s,\xi_s)\sigma_s dW_s, \hspace{.5cm} \text{a.s., for any $t\in I$} 
\end{equation}
\end{lemma}
\begin{proof}
Fix $t\in I$. Since the function $\psi$ is assumed to be of class $\mathscr{C}^{1,2}$, from It\^o's formula (see \cite{brz}, Theorem 2.4), we have that the process $s\mapsto\nabla_2 \psi(s,\xi_s)\sigma_s$, for $s\leq t$, is stochastically integrable and the following representation holds true,
\begin{equation} \label{eq:ito_representation_f}
\psi(t,\xi_t)=\psi(0,\xi_0) +\int_0^t a_s(\psi) ds + \int_0^t \nabla_2 \psi(s,\xi_s)\sigma_s dW_s,  \hspace{1cm} \text{a.s.}
\end{equation}
where, for any $s\leq t$, we set
\begin{equation}\label{eq:a(f)}
a_s(\psi)\triangleq \nabla_1 \psi(s,\xi_s) + \nabla_2 \psi(s,\xi_s) b_s + \frac{1}{2} \text{tr} (\nabla_2^2 \psi(s,\xi_s);\sigma_s).
\end{equation}

Thus, the result follows directly from (\ref{eq:ito_representation_f}) jointly with (\ref{eq:a(f)}), since $\psi$ is assumed to be a BS-function relative to $\xi$.
\end{proof}

We now introduce a class of $E^\ast$-valued processes that we shall use later on in this article.
\begin{definition} \label{de:Phi}
By a \emph{P-set} relative to $\xi$ we understand any set $\Phi$ of $E^\ast$-valued and adapted processes $\phi\triangleq \lbrace \phi_t : t\in I \rbrace$ such that the following conditions are met:
\begin{itemize}
\item[(i)] For any $\phi\in \Phi$, one has that $\Vert \phi \Vert_{\infty}< \infty$, where we set
\begin{equation}
\Vert \phi \Vert_\infty \triangleq \inf\lbrace C \geq 0 :  \Vert \phi_t \Vert_{E^\ast} \leq C \ \text{a.s. for any $t\in I$}\rbrace. \nonumber
\end{equation}
\item[(ii)] For any $\phi\in \Phi$, there exists a function $\varphi: I\times E \rightarrow E^\ast$ of class $\mathscr{C}_b^{1,2}$, such that 
\begin{equation}
\phi_t=\varphi(t,\xi_t), \hspace{1cm} \text{a.s., for any $t\in I$}. \nonumber
\end{equation}
\item[(iii)] For any $\phi\in \Phi$, the following identity holds true a.s.,
\begin{equation}\label{eq:self_financing}
\langle\phi_t,\xi_t\rangle_E = \langle\phi_0,\xi_0\rangle_E + \int_0^t \langle \phi_s,b_s \rangle_Eds + \int_0^t \langle\phi_s,\sigma_s \rangle_E dW_s,
\end{equation}
\end{itemize}
\end{definition}
Notice that in Definition \ref{de:Phi}, given any $\phi \in \Phi$, the variables $\langle\phi_t,\sigma_t\rangle_E$, for $t\in I$, form an adapted $H$-valued process. In this particular case, $\langle\cdot,\cdot\rangle_E$ is thus regarded as the $H$-valued dual pairing between $\gamma(H,E)$ and $E^\ast$. Besides, from the inequality (\ref{eq:ideal_dual}), one has that
\begin{equation}\label{eq:ideal_phi_sigma}
\Vert \langle\phi_t,\sigma_t\rangle_E \Vert_H \leq \Vert \phi_t\Vert_{E^\ast} \Vert \sigma_t \Vert_{\gamma(H,E)}, \hspace{1cm} \text{a.s., for any $t\in I$.}
\end{equation}

\begin{definition}
Fix a P-set $\Phi$ relative to $\xi$ and consider a function $f:I\times E\rightarrow \mathbb{R}$. For any $\phi\in\Phi$, the process $F(\phi)\triangleq \lbrace F_t(\phi) : t\in I \rbrace$ defined by
\begin{equation} \label{eq:financial exposure}
F_t(\phi) \triangleq f(t,\xi_t)- \langle \phi_t,\xi_t\rangle_E, \hspace{1cm} \text{for any $t\in I$},
\end{equation}
is said to be the \emph{discrepancy process} between $f$ and $\phi$ relative to $\xi$.
\end{definition}
%Hence, we write $\bar{F}(\phi)\triangleq \lbrace \bar{F}_t(\phi): t\in I \rbrace$ for the centred process obtained by setting $\bar{F}_t(\phi)\triangleq F_t(\phi)- \mathbb{E}F_t(\phi)$, for any $t\in I$.
If not otherwise specified, where a function $f:I\times E\rightarrow \mathbb{R}$ and a P-set $\Phi$ relative to $\xi$ are fixed, for any $\phi\in \Phi$, we always write $F(\phi)$ to denote the discrepancy process between $f$ and $\phi$ relative to $\xi$. 

%The following lemma provides sufficient conditions for the process (\ref{eq:financial exposure}) to be square integrable.
\begin{lemma} \label{le:integrability_F}
Let $\Phi$ be a P-set relative to $\xi$ and consider a function $f:I\times E\rightarrow \mathbb{R}$. If $f$ is of class $\mathscr{C}_b^{1,2}$, then one has that $F_t(\phi)\in L^2(\Omega)$, for any $t\in I$ and $\phi\in\Phi$, with 
\begin{equation}
\sup_{t\in I}\Vert F_t(\phi)\Vert^2_{L^2(\Omega)} < \infty. \nonumber
\end{equation}
\begin{proof}
First, notice that since $f(\cdot,0)$ is assumed to be continuous on $I$, we have that
\begin{equation}
\Vert f(\cdot,0) \Vert_\infty\triangleq \sup_{t\in I} \vert f(t,0)\vert< +\infty \nonumber.
\end{equation}

Besides, for any $t\in I$, one has that the following inequalities a.s. hold true, 
\begin{eqnarray}
\vert f(t,\xi_t)\vert^2 &\leq &  \Vert \nabla_2 f\Vert^2_\infty \Vert \xi_t \Vert^2_E + \vert f(t,0)\vert^2  \nonumber \\
&\leq & \Vert \nabla_2 f\Vert^2_\infty \Vert \xi_t \Vert^2_E + \Vert f(\cdot,0) \Vert^2_\infty. \nonumber
\end{eqnarray}

Fix $\phi\in \Phi$, and note that $\vert \langle \phi_t , \xi_t \rangle_E \vert \leq \Vert \phi_t \Vert_{E^\ast}\Vert \xi_t \Vert_E$ a.s. for any $t\in I$.
Thus, we obtain, 
\begin{equation}
\sup_{t\in I}\Vert F_t(\phi)\Vert^2_{L^2(\Omega)}\leq ( \Vert \nabla_2 f\Vert^2_\infty+\Vert \phi \Vert^2_\infty)\sup_{t\in I}\Vert \xi_t \Vert_{L^2(\Omega;E)}^2 + \Vert f(\cdot,0) \Vert^2_\infty, \nonumber
\end{equation}
and the result holds true since the function $f$ is assumed to be of class $\mathscr{C}^{1,2}_b$, jointly with the condition (i) in Definition \ref{de:Phi} and Lemma \ref{le:xi_well_defined}.
\end{proof}
\end{lemma}

The following result will play a relevant role later on in this Section.
\begin{lemma}\label{le:optimal_strategy_malliavin}
Let $\Phi$ be a P-set relative to $\xi$ and $f:I\times E\rightarrow \mathbb{R}$ a function of class $\mathscr{C}_b^{1,2}$. One has that $F_t(\phi)\in \mathbb{H}^{1,2}$, for any $t\in I$ and $\phi\in \Phi$, with
\begin{equation}  \label{eq:malliavin_derivative_F}
DF_t(\phi)= (\nabla_2 f(t,\xi_t) - \phi_t) D\xi_t, \hspace{1cm} \text{a.s.}
\end{equation}
\end{lemma} 
It is worth to be highlighted that, for any $t\in I$, the identity (\ref{eq:malliavin_derivative_F}) is to be regarded as
\begin{equation}\label{eq:DF_H_pairing}
DF_t(\phi)=\nabla_2 f(t,\xi_t)D\xi_t - \langle \phi_t,D\xi_t \rangle_E, \hspace{1cm} \text{a.s.} 
\end{equation}
In particular, since $D\xi_t\in L^2(\Omega;\gamma(\mathscr{H},E))$, for any $t\in I$, in the identity (\ref{eq:DF_H_pairing}), we understand $\langle\cdot,\cdot\rangle_E$ as the $\mathscr{H}$-valued dual pairing between $\gamma(\mathscr{H},E)$ and $E^\ast$.
\begin{proof}[Proof of Lemma \ref{le:optimal_strategy_malliavin}]
Here and throughout, we fix $t\in I$. First of all, since the function $f$ is assumed to be of class $\mathscr{C}_b^{1,2}$ and $\xi_t\in \mathbb{H}^{1,2}(E)$ thanks to Lemma \ref{le:xi_malliavin_domain}, the chain rule for the Malliavin derivative (see \cite{pronk}, Proposition 3.8) applies and we get $f(t,\xi_t)\in \mathbb{H}^{1,2}$, with 
\begin{equation} \label{eq:malliavin_derivative_f}
Df(t,\xi_t)= \nabla_2 f(t,\xi_t) D\xi_t, \hspace{1cm} \text{a.s.}
\end{equation} 

% Ito
Fix now $\phi\in\Phi$. According to the assumption (ii) in Definition \ref{de:Phi}, there exists a function $\varphi:I\times E\rightarrow E^\ast$ of class $\mathscr{C}_b^{1,2}$ such that the identity $\phi_t=\varphi(t,\xi_t)$ holds true a.s. 

According to It\^o's formula (see \cite{brz}, Theorem 2.4), we obtain that the process $\nabla_2\varphi(s,\xi_s)\sigma_s$, for $s\leq t$, is stochastically integrable and the following representation holds a.s.,
\begin{equation}
\varphi(t,\xi_t)=\varphi(0,\xi_0)+\int_0^t a_s(\varphi)ds +\int_0^t \chi_s(\varphi) dW_s, \nonumber
\end{equation}
where, for $s\leq t$, we set,
\begin{eqnarray}
a_s (\varphi)&\triangleq &\nabla_1\varphi(s,\xi_s)+ \nabla_2\varphi(s,\xi_s)b_s +\frac{1}{2} \text{tr}(\nabla_2^2 \varphi(s,\xi_s);\sigma_s); \nonumber \\
\chi_s(\varphi) &\triangleq &\nabla_2\varphi(s,\xi_s)\sigma_s. \nonumber
\end{eqnarray} 
%Then, since $E$ is assumed to be a UMD space, also its topological dual $E^\ast$ turns out to be a UMD space as well.

Hence, given any orthonormal basis $h_1,h_2,...$ in $H$, the extension of the It\^o's formula (see \cite{brz}, Corollary 2.6) applied to the dual pairing $\langle \cdot,\cdot\rangle_E$ gives, 
\begin{multline} \label{eq:portfolio_ito_representation}
\langle \varphi(t,\xi_t),\xi_t \rangle_E = \langle \varphi(0,\xi_0),\xi_0\rangle_E + \int_0^t \langle \varphi(s,\xi_s),b_s\rangle_E ds + \int_0^t\langle \varphi(s,\xi_s),\sigma_s \rangle_E dW_s \\ + \int_0^t \langle a_s,\xi_s\rangle_E ds  + \int_0^t \langle \chi_s,\xi_s\rangle_E dW_s + \int_0^t \sum_{n\geq 1} \langle \chi_s h_n, \sigma_sh_n \rangle_E ds, \hspace{1cm} \text{a.s.}
\end{multline}

As a direct consequence, the condition (\ref{eq:self_financing}) jointly with the identity (\ref{eq:portfolio_ito_representation}) implies that $\chi_t=0$ a.s., and in particular 
\begin{equation}\label{eq:frechet_varphi}
\nabla_2\varphi(t,\xi_t)=0, \hspace{1cm} \text{a.s.}
\end{equation}

On the other hand, the chain rule for the Malliavin derivative (see \cite{pronk}, Proposition 3.8) applies and it gives that $\varphi(t,\xi_t)\in \mathbb{H}^{1,2}(E^\ast)$, with
\begin{equation}\label{eq:derivative_phi_zero}
D\varphi(t,\xi_t)=\nabla_2\varphi(t,\xi_t)D\xi_t, \hspace{1cm} \text{a.s.}
\end{equation}
Thus, jointly with the identity (\ref{eq:frechet_varphi}), we get 
\begin{equation}\label{eq:derivative_phi_zero_2}
D\varphi(t,\xi_t)=0, \hspace{1cm} \text{a.s.} 
\end{equation}
The product rule for the Malliavin derivative applied to the pairing $\langle\cdot,\cdot \rangle_E$ (see \cite{pronk}, Lemma 3.6), and jointly with (\ref{eq:derivative_phi_zero_2}) gives that
\begin{eqnarray} \label{eq:malliavin_derivative_varphi}
D\langle \varphi(t,\xi_t),\xi_t \rangle_E &=& \langle D\varphi(t,\xi_t),\xi_t\rangle_E + \langle \varphi(t,\xi_t),D\xi_t \rangle_E, \ \ \ \text{a.s.} \nonumber \\
&=& \langle \varphi(t,\xi_t),D\xi_t \rangle_E \ \ \ \text{a.s.} 
\end{eqnarray}
%Notice that, $D\xi_t\in\gamma (H,E)$ and $\varphi(t,\xi_t)\in E^\ast$, thus on the right hand side of (\ref{eq:malliavin_derivative_varphi}) the operation $\langle \cdot,\cdot\rangle_E$ denotes the $H^\ast$-valued dual pairing between $\gamma(H,E)$ and $E^\ast$.
%
%Notice that from the Cauchy-Schwarz inequality for the dual pairing $\langle\cdot,\cdot\rangle_E$, we get
%\begin{eqnarray}
%\Vert \langle \varphi(t,\xi_t),\xi_t \rangle_E \Vert_{L^2(\Omega)}^2 &\leq & \mathbb{E}\lbrace \Vert \varphi(t,\xi_t) \Vert^2_{E^\ast} \Vert \xi_t \Vert^2_E \rbrace \nonumber \\
%&=& \Vert \varphi \Vert^2_\infty  \Vert \xi_t \Vert^2_{L^2(\Omega, E)}. \nonumber
%\end{eqnarray}
Then, from the inequality (\ref{eq:ideal_dual}) and by means of the natural identification $\gamma(\mathscr{H},\mathbb{R})=\mathscr{H}$, we obtain that
\begin{eqnarray}
\Vert D\langle \varphi(t,\xi_t),\xi_t\rangle_E \Vert^2_{L^2(\Omega;\mathscr{H})} &=& \mathbb{E}\lbrace\Vert \langle \varphi(t,\xi_t),D\xi_t\rangle_E \Vert^2_{\mathscr{H}}\rbrace \nonumber \\
&\leq & \mathbb{E}\lbrace \Vert \varphi(t,\xi_t)\Vert_{E^\ast}^2\Vert D\xi_t \Vert^2_{\gamma(\mathscr{H},E)}\rbrace \nonumber \\
&\leq & \Vert \phi \Vert^2_\infty\Vert D\xi_t \Vert^2_{L^2(\Omega,\gamma(\mathscr{H},E))} \nonumber
\end{eqnarray}
Thus, since $\xi_t\in\mathbb{H}^{1,2}(E)$ due to Lemma \ref{le:xi_malliavin_domain}, we have $\langle \phi_t,\xi_t\rangle_E\in \mathbb{H}^{1,2}$ and hence $F_t(\phi)\in\mathbb{H}^{1,2}$.

Finally, the identity (\ref{eq:malliavin_derivative_F}) follows by the linearity of the Malliavin derivative from equations (\ref{eq:malliavin_derivative_f}) and (\ref{eq:malliavin_derivative_varphi}).
\end{proof}

\begin{definition}\label{de:risk_functional_F}
Let $\Phi$ be a P-set relative to $\xi$. Given a function $f:I\times E\rightarrow \mathbb{R}$ of class $\mathscr{C}_b^{1,2}$, we refer to the functional $\mathcal{F}:\Phi\rightarrow \mathbb{R}$ defined by setting 
\begin{equation}\label{eq:R_risk_functional}
\mathcal{F}(\phi) \triangleq \int_I \mathbb{E}\lbrace\vert F_t(\phi) - \mathbb{E}F_t(\phi) \vert^2\rbrace dt, \hspace{1cm} \text{for any $\phi\in\Phi$,}
\end{equation}
as the \emph{risk functional} relative to $\xi$ induced by $f$ over $\Phi$. 
\end{definition}

If not otherwise specified, where a function $f:I\times E\rightarrow \mathbb{R}$ and a P-set $\Phi$ relative to $\xi$ are fixed, we always write $\mathcal{F}$ to denote the risk functional relative to $\xi$ induced by $f$ over $\Phi$.

The following theorem tells us that the risk functional (\ref{eq:R_risk_functional}) admits two different equivalent representations. 

\begin{theorem}\label{pr:optimal_strategy}
Let $\Phi$ be a P-set relative to $\xi$ and $f:I\times E \rightarrow \mathbb{R}$ a function of class $\mathscr{C}_b^{1,2}$. One has that,
\begin{itemize}
\item[$(i)$] The  functional $\mathcal{F}$ admits the following representation
\begin{equation}\label{eq:optimal_malliavin}
\mathcal{F}(\phi)=\int_I \mathbb{E}\bigg\lbrace\int_0^t\big\Vert \mathbb{E}\lbrace (\nabla_2 f(t,\xi_t) - \phi_t) D_s\xi_t \vert \mathscr{G}^W_s \rbrace \big\Vert_{H}^2 ds\bigg\rbrace dt, \hspace{1cm} \text{for any $\phi\in \Phi$.}
\end{equation}
\item[$(ii)$] In the special case when
 $f$ is assumed to be a BS-function relative to $\xi$ and $b_t=0$ a.s., for any $t\in I$, the functional $\mathcal{F}$ boils down to
\begin{equation}\label{eq:optimal_frechet}
\mathcal{F}(\phi) = \mathbb{E}\bigg \lbrace \int_I\big\Vert (\nabla_2 f(t,\xi_t)-\phi_t)\sigma_t \big\Vert_{H}^2 (1-t)dt\bigg\rbrace, \hspace{1cm} \text{for any $\phi\in \Phi$.} 
\end{equation}
\end{itemize}
\end{theorem}
It is worth to be noted that, since the process $\sigma$ takes values in $\gamma(H,E)$, the right hand side of the identity (\ref{eq:optimal_frechet}) is to be understood as follows $$(\nabla_2 f(s, \xi_s)-\phi_s)\sigma_s\triangleq \nabla_2 f(s,\xi_s)\sigma_s-\langle \phi_s,\sigma_s\rangle_E,$$ where $\langle \cdot,\cdot\rangle_E$ denotes the $H$-valued dual pairing between $\gamma(H,E)$ and $E^{\ast}$. \\

The proof of Theorem \ref{pr:optimal_strategy} is based on the following lemma.
\begin{lemma}\label{le:dfritfless_and_heat_equation}
Let $\Phi$ be a P-set relative to $\xi$ and $f:I\times E \rightarrow \mathbb{R}$ a function of class $\mathscr{C}_b^{1,2}$.
In the particular case when $f$ is assumed to be a BS-function relative to $\xi$ and $b_t=0$ a.s., for any $t\in I$,  we have a.s.
\begin{equation}\label{eq:ito_representation_Ft_driftless}
F_t(\phi)=F_0(\phi) + \int_0^t (\nabla_2f(s,\xi_s)- \phi_s)\sigma_s  dW_s,
\end{equation}
for any $t\in I$ and $\phi\in \Phi$.
\end{lemma}
\begin{proof}[Proof of Lemma \ref{le:dfritfless_and_heat_equation}]
Fix $t\in I$ and note that, since the function $f$ is assumed to be a BS-function relative to $\xi$,  Lemma \ref{le:BL-function} gives that 
\begin{equation}
f(t,\xi_t)=f(0,\xi_0) +\int_0^t \nabla_2 f(s,\xi_s)b_s ds + \int_0^t \nabla_2 f(s,\xi_s)\sigma_s dW_s,  \hspace{1cm} \text{a.s.} \nonumber
\end{equation}

Thus, for any $\phi\in \Phi$, from condition (\ref{eq:self_financing}) we get that the variable $F_t(\phi)$, for $t\in I$, admits the following representation, 
\begin{equation}\label{eq:ito_representation_Ft}
F_t(\phi)=F_0(\phi) + \int_0^t (\nabla_2 f(s, \xi_s)-\phi_s)b_s ds + \int_0^t (\nabla_2f(s,\xi_s)- \phi_s)\sigma_s  dW_s, \ \ \ \text{a.s.}
\end{equation}

Besides, since $b_t=0$ a.s., the representation (\ref{eq:ito_representation_Ft}) boils down to the identity (\ref{eq:ito_representation_Ft_driftless}).
\end{proof}
\begin{proof}[Proof of Theorem \ref{pr:optimal_strategy}]
First, we prove the statement $(i)$.
Since $f$ is assumed to be of class $\mathscr{C}_b^{1,2}$, from Lemma \ref{le:optimal_strategy_malliavin} we get that $F_t(\phi)\in\mathbb{H}^{1,2}$, for any $\phi\in \Phi$ and $t\in I$, and that
\begin{equation}\label{eq:malliavin_derivative_F1}
DF_t(\phi)= (\nabla_2 f(t,\xi_t) - \phi_t) D\xi_t, \hspace{1cm} \text{a.s.}
\end{equation}
Thus, fix $t\in I$ and notice that, since the variable $F_t(\phi)$ is $\mathscr{G}_t^W$-measurable, Clarke-Ocone formula (see \cite{maas}, Theorem 6.6) gives, 
\begin{equation}\label{eq:clarke_ocone_Ft}
F_t(\phi) -\mathbb{E}F_t(\phi) = \int_0^t \mathbb{E}\lbrace D_s F_t(\phi)\vert \mathscr{G}^W_s \rbrace  dW_s \hspace{1cm} \text{a.s.}
\end{equation} 

On the other hand, a direct application of Lemma \ref{le:ito_isometry} gives that,
\begin{equation}\label{eq:ito_isometry_Ft}
%\mathbb{E}\lbrace\vert F_t(\phi) - \mathbb{E}F_t(\phi) \vert^2\rbrace &=&
\mathbb{E}\bigg\lbrace \bigg\vert
\int_0^t \mathbb{E}\lbrace D_s F_t(\phi)\vert \mathscr{G}^W_s \rbrace  dW_s \bigg\vert^2\bigg\rbrace = \mathbb{E}\bigg\lbrace\int_0^t\Vert \mathbb{E}\lbrace D_s F_t(\phi)\vert \mathscr{G}^W_s \rbrace \Vert_{H}^2 ds\bigg\rbrace, 
\end{equation}

Then, when recasting the identity (\ref{eq:ito_isometry_Ft}) in terms of the representation (\ref{eq:malliavin_derivative_F1}), jointly with the identity (\ref{eq:clarke_ocone_Ft}), we obtain 
\begin{equation} \label{eq:ito_isometry_Ft2}
\mathbb{E}\lbrace \vert F_t(\phi) -\mathbb{E}F_t(\phi) \vert^2 \rbrace=\mathbb{E}\bigg\lbrace\int_0^t\Vert \mathbb{E}\lbrace (\nabla_2f(t,\xi_t)-\phi_t)D_s\xi_t\vert \mathscr{G}_s^W\rbrace \Vert_{H}^2 ds\bigg\rbrace.
\end{equation}
As a result, when integrating both the sides of the identity (\ref{eq:ito_isometry_Ft2}) with respect to the variable $t\in I$, we obtain the representation (\ref{eq:optimal_malliavin}).

We now prove the statement $(ii)$. To this purpose, fix $\phi\in\Phi$ and notice that $\mathbb{E}F_t(\phi)=\mathbb{E}F_0(\phi)$ a.s., for any $t\in I$. Thus, since $f$ is assumed to be a BS-function relative to $\xi$ and $b_t=0$, a.s. for any $t\in I$, Lemma \ref{le:dfritfless_and_heat_equation} gives that the following representation holds true a.s.
\begin{equation}\label{eq:centered_F_diffusion}
F_t(\phi) -\mathbb{E}F_t(\phi) = F_0(\phi) -\mathbb{E}F_0(\phi)  + \int_0^t (\nabla_2f(s,\xi_s)- \phi_s)\sigma_s  dW_s, \hspace{1cm} \text{for any $t\in I$.}
\end{equation}
Then, let $p(x)=\vert x \vert^2$, for any $x\in\mathbb{R}$.  According to the representation (\ref{eq:centered_F_diffusion}), for any $t\in I$, notice that a direct application of It\^o's formula (see \cite{brz}, Theorem 2.4) leads to
\begin{multline}\label{eq:ito_representation_Ft^2}
\vert F_t(\phi) -\mathbb{E}F_t(\phi)\vert^2= \vert F_0(\phi) -\mathbb{E}F_0(\phi)\vert^2 \\
+\frac{1}{2}\int_0^t\text{tr}(\nabla^2p(F_s(\phi) -\mathbb{E}F_s(\phi));(\nabla_2f(s,\xi_s)- \phi_s)\sigma_s)ds  + \int_0^t \kappa_s(\phi) dW_s, \hspace{.2cm} \text{a.s.,} 
\end{multline}
where, for any $s\leq t$, we set, 
\begin{equation}\label{eq:alpha_definition}
\kappa_{s}(\phi) \triangleq 2 (F_s(\phi)-\mathbb{E}F_s(\phi))(\nabla_2f(s,\xi_s)- \phi_s)\sigma_s. \nonumber
\end{equation}

Let $h_1,h_2,...$ be an orthonormal basis of $H$ and note that $\nabla^2 p(x)=2$, for any $x\in \mathbb{R}$. Then by definition of the trace operator $\text{tr}(\cdot  ; \cdot)$, we obtain,
\begin{eqnarray}\label{eq:trace=H^ast}
\text{tr}(\nabla^2_2p(F_s(\phi) -\mathbb{E}F_s(\phi));(\nabla_2f(s,\xi_s)- \phi_s)\sigma_s) &=& 2\sum_{n\geq 1} (\nabla_2f(s,\xi_s)- \phi_s)\sigma_sh_n)^2 \nonumber\\ 
 &=& 2\Vert  (\nabla_2f(s,\xi_s)- \phi_s)\sigma_s \Vert^2_{H}. \nonumber
\end{eqnarray}
%since the process $F(\phi)$ is assumed to be a $\mathscr{G}^W$-martingale, for any $t\in I$, we must have $a_t = \langle \varphi(t,\xi_t),b_t\rangle_E$ a.s. and in particular,  
%\begin{equation}\label{eq:derivative_representation}
%F_t(\phi) = F_0(\phi) + \int_0^t (\nabla_2 f(s,\xi_s)-\phi_s) \sigma_s dW_s, \ \ \ \text{a.s.} 
%\end{equation}
%Notice that, for any $t\in I$ one has $\mathbb{E}F_t(\phi)=\mathbb{E}F_0$ and hence a direct application of It\^o isometry (see \cite{car}, equation (4.7) in \S 4.2.1) to (\ref{eq:derivative_representation}) gives,
Notice that $\mathbb{E}\vert F_0(\phi)-\mathbb{E}F_0(\phi)\vert^2=0$, since $F_0(\phi)$ is $\mathscr{G}_0^W$-measurable, and hence $F_0(\phi)=\mathbb{E}F_0(\phi)$ a.s. Thus, from the identity (\ref{eq:ito_representation_Ft^2}) we get,
\begin{equation} \label{eq:ito_isometry_Ft3}
\mathbb{E}\lbrace\vert F_t(\phi) - \mathbb{E} F_t(\phi) \vert^2 \rbrace  
= \mathbb{E}\bigg\lbrace \int_0^t \Vert(\nabla_2 f(s,\xi_s)-\phi_s)\sigma_s \Vert^2_{H} ds \bigg\rbrace.
\end{equation}

Hence, when integrating both the sides of the identity (\ref{eq:ito_isometry_Ft3}) with respect to $t\in I$, we obtain the representation (\ref{eq:optimal_frechet}).
\end{proof}

\begin{definition}\label{de:optimality}
Let $\Upsilon$ be some set and consider a functional $\mathcal{G}:\Upsilon \rightarrow \mathbb{R}$. We call \emph{$\mathcal{G}$-optimal} any element $\upsilon^\ast\in \Upsilon$ that verifies the following inequality
\begin{equation}\label{eq:optimization_problem_mean}
\mathcal{G}(\upsilon^\ast) \leq \mathcal{G}(\upsilon), \hspace{1cm} \text{for any $\upsilon\in \Upsilon$}.
\end{equation}
\end{definition}

\begin{remark}
Fix a P-set $\Phi$ relative to $\xi$ and a function $f:I\times E\rightarrow \mathbb{R}$ of class $\mathscr{C}_b^{1,2}$. 
In the particular case when $f$ is a BS-function relative to $\xi$, the statement (ii) in Theorem \ref{pr:optimal_strategy} tells
 us that a process $\phi^\ast\in \Phi$ is $\mathcal{F}$-optimal if it minimizes the functional
\begin{equation} \label{eq:risk_functinal_beta}
\mathcal{F}(\phi)= \mathbb{E} \bigg\lbrace \int_0^1 \Vert (\nabla_2 f(t,\xi_t) - \phi_t)\sigma_t \Vert_H^2 (1-t) dt \bigg\rbrace, \hspace{1cm} \text{for any $\phi\in \Phi$}. 
\end{equation}

Besides, in the particular case when fixing $x^\ast_1,...,x_n^\ast\in E^\ast$ and assuming that for any $\phi\in \Phi$ there exists $\beta=(\beta_1,...,\beta_n)\in \mathbb{R}^n$ such that  the following representation holds true a.s.
\begin{equation}\label{eq:phi_beta_representaion}
\phi_t = \sum_{i=1}^n \beta_i x_i^\ast, \hspace{1cm} \text{for any $t\in I$}, 
\end{equation}
the minimization of the functional (\ref{eq:risk_functinal_beta}) turns out to be analytically manageable, since it boils down to a quadratic form optimization problem. 

Indeed, when identifying any process $\phi\in \Phi$ with the element $\beta\in \mathbb{R}^n$ that satisfies the identity (\ref{eq:phi_beta_representaion}), \emph{mutatis mutandis} the functional (\ref{eq:risk_functinal_beta}) may be recast as follows
\begin{equation}
\mathcal{F}(\beta)= \sum_{ij=1}^n A_{ij}\beta_i \beta_j -2 \sum_{i=1}^n B_i\beta_i + C, \hspace{1cm} \text{for any $\beta \in \mathbb{R}^n$,} \nonumber
\end{equation}
%\begin{eqnarray}
%R(\beta) &=& \mathbb{E}\bigg\lbrace \int_0^1 \bigg\Vert \nabla_2 f(t,\xi_t)\sigma_t - \sum_{i=1}^N \beta_i \langle x_i^\ast, \sigma_t \rangle_E \bigg\Vert_H^2 (1-t)dt \bigg\rbrace \nonumber\\ 
%&=& \mathbb{E}\bigg\lbrace \int_0^1 \Vert \nabla_2 f(t,\xi_t)\sigma_t \Vert_H^2 (1-t)dt \bigg\rbrace + \mathbb{E}\bigg\lbrace \int_0^1 \bigg\Vert \sum_{i=1}^N \beta_i \langle x_i^\ast, \sigma_t \rangle_E \bigg\Vert_H^2 (1-t)dt \bigg\rbrace \nonumber\\
%&& -2 \mathbb{E}\bigg\lbrace \int_0^1 \bigg\langle \nabla_2 f(t,\xi_t)\sigma_t, \sum_{i=1}^N \beta_i \langle x_i^\ast, \sigma_t \rangle_E \bigg\rangle_H (1-t)dt \bigg\rbrace \nonumber \\
%&=& \sum_{ij=1}^N A_{ij}\beta_i \beta_j + \sum_{i=1}^N B_i\beta_i + C, \nonumber
%\end{eqnarray}
where, for any $i,j=1,...,n$ we set
\begin{eqnarray}
A_{ij} &=& \mathbb{E}\bigg\lbrace \int_0^1 \bigg\langle  \langle x_i^\ast, \sigma_t \rangle_E , \langle x_j^\ast, \sigma_t \rangle_E \bigg\rangle_H (1-t)dt \bigg\rbrace, \nonumber \\
B_i &=& \mathbb{E}\bigg\lbrace \int_0^1 \bigg\langle \nabla_2 f(t,\xi_t)\sigma_t, \langle x_i^\ast, \sigma_t \rangle_E \bigg\rangle_H (1-t)dt \bigg\rbrace, \nonumber \\
 C &=&  \mathbb{E}\bigg\lbrace \int_0^1 \Vert \nabla_2 f(t,\xi_t)\sigma_t \Vert_H^2 (1-t)dt \bigg\rbrace. \nonumber
\end{eqnarray}

As a result, in the special case when the symmetric matrix $A=(A_{ij})_{ij}$ turns out to be positive definite, there is a unique $\mathcal{F}$-optimal element $\beta^\ast\in \Phi$. 
Further, when defining $B\triangleq (B_1,...,B_n)\in \mathbb{R}^n$, the element $\beta^\ast\in \Phi$ is obtained as the solution the following $n$-dimensional inverse problem,
\begin{equation}
A\beta^{\ast} = B. \nonumber
\end{equation}
\end{remark}

\

\section{The optimal portfolio problem.} \label{sec:financial_setup}
%The results we presented in Section \ref{sec:optimlaity} are motivated by the study of suitable optimality conditions aimed at define a portfolio of risky securities with the same performance as given exposure, when contingent restrictions such as policy and budget constraints are permitted. 
%This problem is assessed in terms of the time evolution of the fundamental risk factors associated to the portfolio, since their fluctuation over time leads to changes in the value of the financial securities that they underlie.

In the sequel, we discuss how the results presented in Section \ref{sec:optimlaity} may be proposed to address the problem of substituting a financial exposure by some constrained portfolio, without misrepresenting its performance in terms of the related inherent risk structure.

Here and in the sequel of the paper, we always assume that the risk factors of a given portfolio are represented by tradable observables. These include the market price of a stock or a commodity itself or some major market benchmark assessing the value of an entire class of securities, such as the interest rate term structure when dealing with the fixed income market.
In this respect, we only refer to the market risk that affects the financial exposure. No other types of risk are considered.

Any element of $E$ represents the overall discounted value of the risk factors at a certain time and we regard $I$ as the reference time interval. For sake of simplicity, we suppose $I$ to define the period of one year, and any of its fractions to be assessed according to a certain day count convention.
\begin{example}
To fix the ideas, when $E$ is chosen to be the Euclidean space of some finite dimension $n\in \mathbb{N}$, then the components of any element $x=(x_1,...,x_n)\in E$ may represent the discounted market prices of $n$ assets at a certain time. Besides, one may choose $E$ to be some space of continuous curve and interpret any of its elements as the structure of the discounted price curve at a certain time.
\end{example} 

We suppose the process (\ref{eq:diffusive_dynamics_general}) to provide a dynamics for the overall discounted value of the risk factors. Moreover, we regard any function $f:I\times E\rightarrow \mathbb{R}$ of class $\mathscr{C}_b^{1,2}$ as a fixed financial exposure and we understand the variable $f(t,\xi_t)$ as its discounted value at any time $t\in I$.

For a fixed P-set $\Phi$ relative to $\xi$, we will interpret any $\phi\in \Phi$ as the dynamics of a certain portfolio of risk factors managed by the trader. Hence, we understand the variable $\langle \phi_t,\xi_t\rangle_E$ as its discounted value at time $t\in I$. 
Notice that this term depends on both the overall discounted value of the risk factors $\xi_t$ and the portfolio composition $\phi_t$ chosen by the investor.
The assumption (ii) in Definition \ref{de:Phi} is not merely technical, since the strategy considered by a rational investor at a certain time should depend on the evolution of the reference market. 
On the other hand, for our applications, the identity (iii) in Definition \ref{de:Phi} generalizes the well-known self-financing condition.

We regard the risk functional $\mathcal{F}$ relative to $\xi$ induced by $f$ over $\Phi$ as the error that occurs when substituting the exposure represented by $f$ by some portfolio within $\Phi$. Such an error is assessed in terms of the average changes of the difference between the exposure and the selected portfolio, due to the fluctuation of the underlying risk factors within the period $I$.
In this respect, a portfolio $\phi^\ast\in\Phi$ turns out to be $\mathcal{F}$-optimal when it provides the best representation of the inherent risk of the exposure $f$, among all the possible choices within $\Phi$.

Theorem \ref{pr:optimal_strategy} shows two different formulations of the functional $\mathcal{F}$ in terms of the operator $\nabla_2 f$, which gauges the sensitivity of the financial exposure to little variations of the underlying risk factors. 
In this respect, the optimization of the functional $\mathcal{F}$  may be regarded as a notion of portfolio immunization via sensitivity analysis.

\

\noindent{\emph{Market valued exposure.}}
It is worth to be noted that in the particular case when $f$ is assumed to be a BS-function relative to $\xi$ and $b_t=0$ a.s., for any $t\in I$, Lemma \ref{le:BL-function} assures that 
\begin{equation}
\mathbb{E}f(t,\xi_t)=\mathbb{E} f(0,\xi_0), \hspace{1cm} \text{for any $t\in I$.} \nonumber
\end{equation}
As a consequence, we may regard Definition \ref{de:BS-function}
in terms of a risk-free condition. Stated differently, a function $f$ turns out to be a BS-function relative to $\xi$ when the financial exposure that it represents turns out to be market valued.

On the other hand, under the same hypothesis Lemma \ref{le:dfritfless_and_heat_equation} leads to
\begin{equation}
\mathbb{E}F_t(\phi)=\mathbb{E} F_0(\phi), \hspace{1cm} \text{for any $t\in I$ and $\phi\in\Phi$.} \nonumber
\end{equation}
Hence, in the particular case when $f(0,\xi_0)=\langle \phi_0,\xi_0 \rangle_E$ a.s., for any $\phi\in \Phi$, one has $\mathbb{E}F_t(\phi)=0$ a.s. for any $t\in I $ and $\phi\in \Phi$, and the risk functional $\mathcal{F}$ induced by $f$ over $\Phi$ boils down to
\begin{equation}
\mathcal{F}(\phi)=
\Vert F(\phi) \Vert^2_{L^2(\Omega\times I)}, \hspace{1cm} \text{for any $\phi\in \Phi$.} \nonumber
\end{equation} 

This means that when the financial exposure is market valued and it is perfectly hedged by any portfolio $\phi\in\Phi$ at time $t=0$, the $\mathcal{F}$-optimal portfolio $\phi^\ast\in\Phi$ is the one that minimizes the average squared discrepancy with the exposure over time, based on the underlining risk factors fluctuation.

%\begin{remark}
%If there exists a version $\phi^\ast\in \Phi$ of the process $\nabla_2 f(t,\xi_t)$, for $t\in I$, then $\phi^\ast$ turns out to be $F$-optimal. The variable $\nabla_2 f(t,\xi_t)$ gauges the sensitivity of the financial exposure at time $t\in I$ to little perturbations of the underlying risk factors $\xi$. 
%In this respect, the optimization problem (\ref{eq:optimization_problem_mean}) may be regarded as a notion of \emph{portfolio immunization}, and
%depending on the specific case, the variable $\nabla_2 f(t,\xi_t)$, for $t\in I$, may represent: 
%\begin{itemize}
%\item[\emph{$i.$}] the \emph{delta} of a derivative portfolio, when (\ref{eq:diffusive_dynamics_general}) describes  the value of the underlying asset over time;
%\item[\emph{$ii.$}] The \emph{duration} of a fixed income portfolio, when (\ref{eq:diffusive_dynamics_general}) describes the fluctuation of the interest rate term structure over time;
%\item[\emph{$iii.$}] The \emph{vega} of a portfolio, when (\ref{eq:diffusive_dynamics_general}) provides a stochastic model for the volatility of the underlying asset.
%\end{itemize}

\

\section{Constrained hedging and residual risk.} \label{sec:examples}

Let $f:I\times E \rightarrow \mathbb{R}$ be a function of class $\mathscr{C}_b^{1,2}$ and fix a P-set $\Phi$ relative to $\xi$. 
As a consequence of Theorem \ref{pr:optimal_strategy}, if there exists a process $\phi^\ast\in \Phi$ such that the following identity holds true a.s.
\begin{equation}\label{eq:sensitivity_hedging}
\phi_t^\ast = \nabla_2 f(t,\xi_t) \hspace{1cm} \text{for any $t\in I$,}
\end{equation}
then $\phi^\ast$ turns out to be $\mathcal{F}$-optimal, since $\mathcal{F}(\phi^\ast)=0$ in this case. This fact is consistent with the sensitivity-based hedging approach for portfolio immunization.

We discuss this interpretation later on.

\

\noindent{\emph{Sensitivity-based hedging.}}
Assume $E$ to coincide with the Euclidean space $\mathbb{R}^2$. Hence, write  $\langle \cdot,\cdot\rangle_E$ to denote the standard Euclidean product on it and let $e=(e_1,e_2)$ be its canonical basis. On the other hand, we assume $H$ to coincide with $\mathbb{R}$ and thus we regard the $H$-cylindrical process $W$ as a standard one-dimensional Brownian motion. 

Given this framework, note that the space $\gamma(H,E)$ boils down to $E$ itself. Thus, we may represent the process $\sigma\in L^2(\Omega,L^2(I,E))$ in terms of its components by writing $\sigma=(\sigma_{1},\sigma_{2})$, where, for any $i=1,2$, the real valued process $\sigma_i=\lbrace \sigma_{i,t}:t\in I \rbrace$ is defined as follows,
\begin{equation}
\sigma_{i,t}\triangleq\langle e_i, \sigma_t \rangle_E, \hspace{1cm} \text{for any $t\in I$}. \nonumber
\end{equation}

Similarly, we write $\xi_i\triangleq \lbrace \xi_{i,t}:t\in I \rbrace$, for $i=1,2$, to denote the components of the process $\xi$ that are defined by setting
\begin{equation} \label{eq:components_xi}
\xi_{i,t}\triangleq \langle e_i,\xi_t\rangle_E, \hspace{1cm} \text{for $t\in I$}.
\end{equation}
Throughout, we assume that a.s. $b_t=0$ and $\sigma_{2,t}=0$, for any $t\in I$. Moreover, for sake of simplicity, we set $\xi_{2,0}=1$ a.s.

Then, since for $i=1,2$ the following equality a.s. holds true (see \cite{vanweis}, Theorem 4.2),
\begin{equation}
\bigg\langle e_i , \int_0^t \sigma_s dW_s \bigg\rangle_E = \int_0^t \langle  e_i , \sigma_s \rangle_E dW_s, \hspace{1cm} \text{for any $t\in I$,} \nonumber
\end{equation}
we obtain the following representation for the dynamics of the components given in (\ref{eq:components_xi}),
\begin{equation}\label{eq:components_xi_hedging}
\xi_{1,t} = \xi_{1,0} + \int_0^t \sigma_{1,s} dW_s, \ \  \text{a.s.}, \hspace{1cm} \xi_{2,t} = 1, \ \  \text{a.s.}
\end{equation}

We fix a function $f:I\times E\rightarrow \mathbb{R}$ of class  $\mathscr{C}_b^{1,2}$ and a P-set $\Phi$ relative to $\xi$.
For any process $\phi\in \Phi$, we regard the variable $\phi_t$, for $t\in I$, in terms of its components defined for $i=1,2$ by the following identity,
\begin{equation}
\phi_{i,t}\triangleq \langle e_i, \phi_t \rangle_E, \hspace{1cm} \text{for any $t\in I$.} \nonumber
\end{equation}

\begin{proposition}
Let $\Phi$ be a P-set relative to $\xi$. Assume $f$ to be a BS-function relative to $\xi$. If the components of the process $\xi$ are given by (\ref{eq:components_xi_hedging}), then 
\begin{equation}\label{eq:optimization_delta_hedging}
\mathcal{F}(\phi)=\mathbb{E}\bigg \lbrace \int_I \vert (\partial_{x_1} f(t,\xi_{1,t},\xi_{2,t})-\phi_{1,t})\sigma_{1,t} \vert^2 (1-t)dt\bigg\rbrace, \hspace{1cm} \text{for any $\phi\in \Phi$.} 
\end{equation}
\end{proposition}
\begin{proof}
Notice that for any $t\in I$ and $x=(x_1,x_2)\in E$, we may regard $\nabla_2 f(t,x)$ as an element of $E$ with components, 
\begin{equation}
\nabla_2 f(t,x) = (\partial_{x_1} f(t,x_1,x_2),\partial_{x_2} f(t,x_1,x_2)), \nonumber
\end{equation} 
and thus we may set
\begin{equation}
\nabla_2 f(t,x)y \triangleq \langle\nabla_2 f(t,x), y\rangle_E, \hspace{1cm} \text{for any $y\in E$.} \nonumber
\end{equation}

Then, since according to the identities (\ref{eq:components_xi_hedging}) we have $b_t=0$ a.s., for any $t\in I$, and $f$ is a BS-function relative to $\xi$, the result follows directly from the statement (ii) in Theorem \ref{pr:optimal_strategy}, by noting that $\sigma_{2,t}=0$ a.s., for any $t\in I$.
\end{proof}

We may understand the first component of the process $\xi$ as a risk-neutral dynamics for the discounted price of some risky asset. Besides, we regard its second component as a risk-neutral model for the discounted value of the bank account.
The function $f$ represents an European contingent claim written on the risky asset and the P-set $\Phi$ relative to $\xi$ stands for the entire class of the hedging  portfolios.
%Moreover, we understand the variable $\phi_{1,t}$ as the amount of the risky asset within the portfolio $\phi\in \Phi$, while the process $\phi_2$ is understood as the amount of money that is stored as a deposit.

It is worth to be noted that in the particular case when there exists a process $\phi^\ast\in\Phi$ such that the following equality holds true a.s., 
\begin{equation} \label{eq:delta-hedging-condition}
\phi_{1,t}^\ast=\partial_{x_1} f(t,\xi_t), \hspace{1cm} \text{for any $t\in I$,}
\end{equation}
then the process $\phi^\ast$ turns out to be $\mathcal{F}$-optimal, since $\mathcal{F}(\phi^\ast)=0$ thanks to the representation (\ref{eq:optimization_delta_hedging}). In particular, the identity (\ref{eq:delta-hedging-condition}) corresponds to the so-called delta-hedging condition for the contingent claim $f$.
\begin{remark}
Assume that the first component of the process $\phi^{\ast}\in \Phi$ satisfies the identity (\ref{eq:delta-hedging-condition}). Note that, if the following condition holds a.s.
\begin{equation}\label{eq:perfect_headge_t=0}
f(0,\xi_0)=\langle \phi_0^\ast,\xi_0 \rangle_E, 
\end{equation}
one has that $F_0(\phi^{\ast})=0$ a.s.
The identity (\ref{eq:perfect_headge_t=0}) tell us that the financial exposure is perfectly hedged by the portfolio $\phi^\ast\in\Phi$ at time $t=0$. \\
In this particular case, Lemma \ref{le:dfritfless_and_heat_equation} gives that $F_t(\phi^{\ast})=0$ a.s., for any $t\in I$, and the second component of the process $\phi^\ast$ is thus implicitly determined by the following identity
\begin{equation}
\phi_{2,t}^{\ast} = f(t,\xi_t) - \phi_{1,t}^{\ast}\xi_{1,t}\hspace{1cm} \text{a.s., for any $t\in I$.} \nonumber
\end{equation}

Moreover, it is worth to be noted that within the present setup the identity (\ref{eq:heat_equationo_zeta}) boils down to
\begin{equation}\label{eq:bl}
\nabla_1 f(t,\xi_{1,t},\xi_{2,t}) + \frac{1}{2}  \partial_{x_1x_1}  f(t,\xi_{1,t},\xi_{2,t}) \sigma_{1,t}^2= 0, \ \ \text{a.s., for any $t\in I$}, \nonumber
\end{equation}
which corresponds to a Black-Scholes type equation, whereby the risk-free rate is set to be null at any time. 
\end{remark}

%
%Let $r\triangleq \lbrace r_s : t\in I \rbrace$ be some adapted process for the short rate dynamics and $\beta\triangleq \lbrace \beta_t : t\in I \rbrace$ be the unit rolled-up money account process defined by setting
%\begin{equation}
%\beta_t\triangleq \exp{\bigg\lbrace\int_0^t} r_s ds \bigg\rbrace.
%\end{equation}
%Thus, for any $t\in I$, assume the variable $$f(t,\xi_t)\triangleq \pi(t,\beta_t \xi_{1,t})\beta_t^{-1},$$ for some $\pi:I\times\mathbb{R}\rightarrow \mathbb{R}$, to assess the discounted value of a European contingent claim written on the risky asset. Then, condition (\ref{eq:heat_equation}) boils down to 
%\begin{equation}\label{eq:bl}
%\nabla_1 f(t,\xi_t) + r_t\xi_t^{(1)}\nabla_2 f(t,\xi_t) + \frac{1}{2} \sigma_t^{(1)^2}\nabla_2^2  f(t,\xi_t) = r_t  f(t,\xi_t), \ \ \text{a.s., for any $t\in I$}, \nonumber
%\end{equation}
%which corresponds to the Black-Scholes equation. 
%\

\

\noindent{\emph{Correlation and residual risk.}} Throughout, we analyse the case when the set $\Phi$ is defined is such a way that condition (\ref{eq:sensitivity_hedging}) may not be fulfilled. This case turns out to be appealing when the financial exposure and the replication portfolio actually depend upon two different but correlated risk factors.

Assume both $E$ and $H$ to coincide with the Euclidean space $\mathbb{R}^2$. Hence, write  $\langle \cdot,\cdot\rangle_E$ to denote the standard Euclidean product on it and let $e=(e_1,e_2)$ be its canonical basis.
Further, fix a constant $\varrho\in (-1,1)$ and assume the inner product $\langle\cdot,\cdot\rangle_H$ on $H$ to be defined in such a way that $\langle e_i,e_j \rangle_H=1$, if $i=j$, and $\langle e_i,e_j \rangle_H=\varrho$ otherwise.

Let $f:I\times E\rightarrow \mathbb{R}$ be the function defined by setting
\begin{equation}\label{eq:function_f_correlation}
f(t,x)= \langle e_1, x\rangle_E, \hspace{1cm} \text{for any $t\in I$ and $x\in E$.} 
\end{equation}
Notice that $f$ is of class $\mathscr{C}_b^{1,2}$, with $\nabla_2 f(t,x)=e_1$, for any $t\in I$ and $x\in E$. 

Let $\Phi$ be a P-set relative to $\xi$ and assume that for any $\phi\in \Phi$ there exists a real-valued process $\phi_2\triangleq \lbrace \phi_{2,t} :t\in I \rbrace$ such that the following identity holds true 
\begin{equation}\label{eq:phi_bar_representation}
\phi_t= \phi_{2,t}e_2, \hspace{1cm} \text{a.s., for any $t\in I$.}
\end{equation}
Hence, we identify any $\phi\in \Phi$ with the process $\phi_{2}$ that satisfies the identity (\ref{eq:phi_bar_representation}).

\begin{proposition}\label{pr:risk_functional_correlation}
Let $f$ be the function given by the identity (\ref{eq:function_f_correlation}) and $\Phi$ a P-set relative to $\xi$ that verifies the condition (\ref{eq:phi_bar_representation}). If $b_t=0$ a.s., for any $t\in I$, then
\begin{equation}\label{eq:risk_functional_correlation}
\mathcal{F}(\phi) = \mathbb{E}\bigg \lbrace \int_I \Vert \langle e_1, \sigma_t \rangle_E - \phi_{2,t} \langle e_2,\sigma_t\rangle_E \Vert^2_H (1-t)dt\bigg\rbrace, \hspace{1cm} \text{for any $\phi\in\Phi$.}
\end{equation}
\begin{proof}
First of all, notice that  $\nabla_1 f(t,x)=0$ and $\nabla_2^2 f(t,x)=0$, for any $t\in I$ and $x\in I$. Hence, the function  $f$ appears to be a BS-function relative to $\xi$.   
Moreover, note that 
\begin{equation}
\nabla_2 f(t,\xi_t)\sigma_t = \langle e_1 , \sigma_t \rangle_E, \hspace{1cm} \text{for any $t\in I$.} \nonumber
\end{equation}
and for any $\phi\in \Phi$, and that
\begin{equation}
\langle \phi_t , \sigma_t \rangle_E = \phi_{2,t} \langle e_2,\sigma_t\rangle_E, \hspace{1cm} \text{for any $t\in I$}. \nonumber
\end{equation}

Then, the result follows directly from  Theorem \ref{pr:optimal_strategy}.
\end{proof}
\end{proposition}

Notice that, since it is not possible to find a version of the process $\nabla_2 f(t,\xi_t)$, for $t\in I$, that belongs to $\Phi$, the condition (\ref{eq:sensitivity_hedging}) may not be recovered. 

Besides, the following result provides an explicit characterization of a $\mathcal{F}$-optimal process in this case.

\begin{proposition}
Let $f$ be the function given by the identity (\ref{eq:function_f_correlation}) and $\Phi$ a P-set relative to $\xi$ that verifies the condition (\ref{eq:phi_bar_representation}). If $b_t=0$ a.s., for any $t\in I$, then the process $\phi^\ast\in \Phi$ defined by 
\begin{equation}\label{eq:phi_correlation_replication}
\phi_{2,t}^\ast=\varrho (\sigma_{1,t}/\sigma_{2,t}), \hspace{1cm} \text{for any $t\in I$},
\end{equation}
turns out to be $\mathcal{F}$-optimal, where for $i=1,2$ the real valued process $\sigma_i\triangleq \lbrace \sigma_{i,t}:t\in I \rbrace$ is such that 
\begin{equation}\label{eq:sigma_spectre}
\langle e_i, \sigma_t \rangle_E = \sigma_{i,t}e_i, \hspace{1cm} \text{for any $t\in I$.}
\end{equation}
\end{proposition}
\begin{proof}
Notice that Proposition \ref{pr:risk_functional_correlation} applies and the functional $\mathcal{F}$ admits the representation given in (\ref{eq:risk_functional_correlation}).  

On the other hands, a direct computation shows for any $t\in I$, one has
\begin{equation}\label{eq:norm_correlation}
\Vert \langle e_1, \sigma_t \rangle_E - \phi_{2,t} \langle e_2,\sigma_t\rangle_E \Vert^2_H = \sigma_{1,t}^2 - 2\varrho \phi_{2,t} \sigma_{1,t}\sigma_{2,t} + \phi_{2,t}^2\sigma_{2,t}^2,
\end{equation}
which consists in a quadratic form in terms of $\phi_{2,t}$, that attends its minimum at
\begin{equation}
\phi_{2,t}=\varrho (\sigma_{1,t}/\sigma_{2,t}). \nonumber
\end{equation}
Thus, when letting $t$ run over $I$, the process $\phi^\ast\in \Phi$ defined by identity (\ref{eq:phi_correlation_replication}) satisfies $\mathcal{F}(\phi^\ast)\leq \mathcal{F}(\phi)$, for any $\phi\in\Phi$, and hence it turns out to be $\mathcal{F}$-optimal.
\end{proof}

Notice that the $H$-Wiener process $W$ may be understood as a $2$-dimensional Wiener process, whose components $W_i$, for $i=1,2$, are defined by setting,
\begin{equation}\label{eq:components_tew_dimensional}
W_{i,t}\triangleq W_te_i, \hspace{1cm} \text{for any $t\in I$,}
\end{equation} 
with $\varrho$ as their instantaneous correlation, since
\begin{equation}
\mathbb{E}\lbrace W_te_1 \cdot W_te_2 \rbrace = \langle e_1, e_2 \rangle_H t = \varrho t, \hspace{1cm} \text{for any $t\in I$.} \nonumber
\end{equation}

Thus, we may regard the components of the process $\xi$ as a dynamics of the discounted prices of two correlated risky assets. 

In this case, a perfect hedge may not be recovered. Indeed, notice that given $\phi^\ast\in \Phi$ as defined by the identity (\ref{eq:phi_correlation_replication}), the quantity $\mathcal{F}(\phi^\ast)$ is strictly positive for $\varrho < 1$ and it vanishes for $\varrho=1$, which is the particular case when the two assets appear to be completely correlated.

Then, we may regard $\mathcal{F}(\phi^\ast)$ as the residual hedging risk.    
\

\section{Interest rates securities portfolios.}\label{sec: Interest Rates}
In this section we take a close look at the problem of seeking an optimal representation of some fixed income portfolio by considering a portfolio of zero coupon bonds, which are those contracts that pay one unit of a certain currency at some maturity future date. 
%When dealing with such a framework, the main source of risk is provided by the fluctuation of the interest rate term structure over time.

We make the assumption to deal with idealized bonds that are unaffected by credit risk. i.e. the payment at the maturity date is always realized by the issuer of the bond.
We fix $\mathscr{T}\triangleq (1,+\infty)$. Thus, we suppose that there exists a market valued bond maturing at any future time $T\in \mathscr{T}$, and we write $p_t(T)$ to denote its risk-neutral discounted price at time $t\in I$.
Moreover, we refer to the function $T\in \mathscr{T}\mapsto p_t(T)$ as the discounted price curve at time $t\in I$.
It is worth to be highlighted that we let the variable $T$ to run over $\mathscr{T}$, since we only deal with those bonds that do not expire within the first year.

\ 

\noindent{\emph{Price curve dynamics}.}
Throughout, we consider a UMD Banach space $E$ with type $2$ that is represented by some space of continuous and real-valued functions defined on $\mathscr{T}$. We also assume the evaluation functional $\delta_T$ to be continuous and bounded on $E$, for any $T\in \mathscr{T}$.
On the other hand, in order to capture the feature of any maturity specific-risk, it may be reasonable to let $H$ be some infinite dimensional separable Hilbert space and hence $W$ to be a $H$-cylindrical process as defined in Section \ref{sec_preliminaries}.

An instance of this setup is the one suggested by Carmona and Tehranchi \cite{carteh}, by considering the space introduced in the following definition. 

\begin{definition}\label{de:H_w}
Let $w:\mathscr{T}\rightarrow \mathbb{R}^+$ be a positive and increasing function.
We write $\mathcal{H}_w$ to denote the space of absolutely continuous functions $x:\mathscr{T}\rightarrow \mathbb{R}$ with $x(s)\rightarrow 0$, as $s\rightarrow +\infty$, and such that 
\begin{equation}
\int_\mathscr{T} x'(s)^2w(s)ds < +\infty, \nonumber
\end{equation} 
where $x'$ stands for the weak derivative of $x$.
\end{definition}
When endowed with the norm 
\begin{equation}
\Vert x \Vert_{\mathcal{H}_w}\triangleq \bigg\lbrace x(1)^2 + \int_\mathscr{T} x'(s)^2w(s)ds\bigg\rbrace^{1/2},\hspace{1cm} \text{for any $x\in \mathcal{H}_w$,} \nonumber
\end{equation}
the space $\mathcal{H}_w$ turns out to be a Hilbert space, as reported in the Lemma below.
 In this respect, we recall that any Hilbert space is also a UMD Banach space with type $2$.
\begin{lemma}\label{le:condition_w_H}
Let $w:\mathscr{T}\rightarrow \mathbb{R}^+$ be a positive and increasing function such that 
\begin{equation}\label{eq:condition_w_H}
\int_\mathscr{T} w(s)^{-1} ds < +\infty
\end{equation}
then $\mathcal{H}_w$ is a separable Hilbert space and the evaluation functional $\delta_T$ is continuous and bounded on $\mathcal{H}_w$, for any $T\in\mathscr{T}$.
\end{lemma}
\begin{proof}
See, e.g.,  Proposition 6.3 in \cite{car}.
\end{proof}

Here and in the sequel, we regard any element of $E$ as the possible structure of the discounted price curve at a certain time. More precisely, we assume the risk-neutral dynamics of the discounted price curve to be governed by a certain adapted $E$-valued process $p=\lbrace p_t: t\in I \rbrace$. In this respect, let $\sigma=\lbrace \sigma_t : t\in I \rbrace$ be an adapted $H$-strongly measurable process such that $\sigma \in \mathbb{H}^{2,2}(L^2(I;\gamma(H,E)))$, as defined as in Section \ref{sec_preliminaries}, and assume $p_0\in \mathbb{H}^{1,2}(E)$ to be some strongly $\mathscr{G}_0^W$-measurable random variable. Then, we set
\begin{equation}\label{eq:discounted_price_dynamics}
p_t = p_0 + \int_0^t \sigma_s dW_s, \hspace{1cm} \text{for any $t\in I$.}
\end{equation}

Besides, let $\delta_T\in E^\ast$ be the evaluation functional at $T\in \mathscr{T}$ and notice that 
\begin{equation}
p_t(T) = \langle \delta_T , p_t \rangle_E, \hspace{1cm} \text{for any $t\in I$.} \nonumber
\end{equation}  
Thus, we may interpret $\delta_T$ as a portfolio composed by a zero coupon bond expiring at time $T\in \mathscr{T}$, and since the following equality a.s. holds true (see \cite{vanweis}, Theorem 4.2),
\begin{equation}
 \bigg\langle \delta_T, \int_0^t \sigma_s dW_s \bigg\rangle_E=  \int_0^t \langle\delta_T,\sigma_s\rangle_E dW_s, \hspace{1cm}\text{for any $t\in I$},\nonumber
\end{equation}
the identity (\ref{eq:discounted_price_dynamics}) leads to the representation of the risk-neutral dynamics for the discounted price of the bond expiring at $T$ given by 
\begin{equation}\label{eq:T-Bond_dynamics}
p_t(T) = p_0(T) + \int_0^t \sigma_s(T)dW_s, \hspace{1cm} \text{for any $t\in I$,}
\end{equation} 
where for notation simplicity we set $\sigma_t(T)\triangleq\langle\delta_T,\sigma_s\rangle_E$.

\

\noindent{\emph{Risk functional and portfolio duration}.}
Let $f:I\times E\rightarrow \mathbb{R}$ be a function of class $\mathscr{C}_b^{1,2}$ and fix a P-set $\Phi$ relative to the process $p$ given by the identity (\ref{eq:discounted_price_dynamics}). We assume that for any $\phi\in\Phi$ there exists some $T\in \mathscr{T}$, such that 
\begin{equation} \label{eq:phi_delta_bond}
\phi_t=\alpha(T)\delta_T, \hspace{1cm} \text{a.s., for any $t\in I$,} 
\end{equation}
where we set,
\begin{equation}\label{eq:condition_alpha}
\alpha(T) = p_0(T)^{-1} f(0,p_0).
\end{equation} 
As a direct result, we may write $\Phi=\mathscr{T}$ by identifying any process $\phi\in \Phi$ with $T\in \mathscr{T}$ such that the condition (\ref{eq:phi_delta_bond}) is satisfied. 

%The following proposition provides a representation for the risk functional $\mathcal{F}$ induced by $f$ over $\Phi$. 

\begin{proposition}
Let $f:I\times E\rightarrow \mathbb{R}$ be a function of class $\mathscr{C}_b^{1,2}$ and the process $p$ be given by (\ref{eq:discounted_price_dynamics}). Let $\Phi$ be a P-set relative to $p$, whose elements satisfy the identity (\ref{eq:phi_delta_bond}). If $f$ is a BS-function relative to $p$, then
\begin{equation}\label{eq:optimization_T}
\mathcal{F}(T) = \mathbb{E}\bigg \lbrace \int_I\big\Vert (\nabla_2 f(t,p_t)\sigma_t -\alpha(T)\sigma_t(T) \big\Vert_{H}^2 (1-t)dt\bigg\rbrace, \hspace{1cm} \text{for any $T\in\mathscr{T}$.}
\end{equation}
\begin{proof}
The result follows directly form the statement (ii) in Theorem \ref{pr:optimal_strategy}. Indeed, identify the process $p$ with
$\xi$, and notice that in the present case one has $b_t=0$ a.s., for any $t\in I$. 
\end{proof}
\end{proposition}

We may understand the variable $f(t,p_t)$ as the discounted value of a certain interest rate securities portfolio at time $t\in I$. 
We interpret any $\phi\in \Phi$ as the dynamics of a certain portfolio composed by a single bond relative to a fixed maturity and nominal given by the identity (\ref{eq:condition_alpha}). Thus, we suppose the variable $\langle \phi_t,p_t \rangle_E$ to assess its risk-neutral discounted value at time $t\in I$. 

It is worth to be noted that the condition (\ref{eq:condition_alpha}) is reasonable, since it guarantees that 
\begin{equation}
f(0,p_0) = \langle \phi_0,p_0 \rangle_E, \nonumber
\end{equation}
and hence that any $\phi\in \Phi$ provides a perfect hedge relative to the portfolio represented by $f$ at time $t=0$.

We may regard the $E^\ast$-valued variable $\nabla_2 f(t,p_t)$ as a notion of duration of the portfolio $f$ at time $t\in I$, since it may represent the sensitivity of the portfolio to little changes in the structure of the price curve at time $t\in I$. In this respect, notice that if $f$ is a BS-function relative to $p$ and there exists $T^\ast\in \mathscr{T}$ such that 
\begin{equation}
\nabla_2 f(t,p_t)= \alpha(T^\ast)\delta_{T^\ast} \hspace{1cm} \text{a.s., for any $t\in I$,} \nonumber
\end{equation} 
then the representation (\ref{eq:optimization_T}) gives $\mathcal{F}(T^\ast)=0$, and hence $T^\ast$ turns out to be $\mathcal{F}$-optimal. 

In this respect, when the financial exposure $f$ is market valued, we may regard any $\mathcal{F}$-optimal $T^\ast\in\mathscr{T}$ as a notion of the related duration.

\

\begin{figure}
    \centering
        \includegraphics[width=0.58\textwidth]{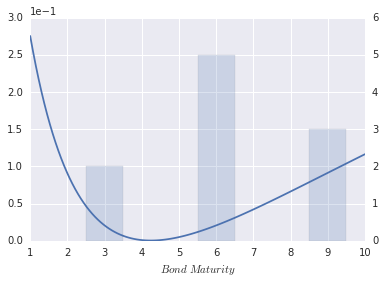}
    \caption{Visual representation of the functional (\ref{eq:blue_line_bond}) when considering the model (\ref{eq:g++2_factors}) for the short-rate dynamics. More precisely, the line represents the functional (\ref{eq:blue_line_bond}) for any bond maturity $T\in (1,10)$, the value of which is assessed by the left side vertical axis. The right side vertical axis assesses the nominals of the bonds within the fixed portfolio, that are represented by the vertical bars in the figure. The model (\ref{eq:g++2_factors}) has been simulated with the parameters $a_1=0.12$, $a_2=0.1$, $\sigma_1=0.16$ and $\sigma_2=0.15$. For sake of simplicity, we considered the case $\varphi(t)=\varphi_1 + \varphi_2 t$, for any $t\in I$, where  $\varphi_1=0.01$ and $\varphi_2=0.15$. Finally, we set $\varrho=-0.01$.} \label{fig:bond_replication}
\end{figure}

%\begin{SCfigure}
%  \centering
%  \caption{$a_1=0.12$ $a_2=0.1$ $\sigma_1=0.16$ $\sigma_2=0.15$ $\varphi(t)=\varphi_1 + \varphi_2 t$ $\varphi_1=0.01$ $\varrho=-0.01$}
%  \includegraphics[width=0.4\textwidth]%
%    {g_plusplus1.png}% picture filename
%    \label{fig:bond_replication}
%\end{SCfigure}
%
%

\noindent{\emph{Numerical example}.}
Figure \ref{fig:bond_replication} provides a visual representation of the functional (\ref{eq:optimization_T}), when considering a fixed portfolio constituted by three bonds with different nominals and maturities, and whose risk-neutral discounted value at time $t\in I$ is given by
\begin{equation}
f(t,p_t) = \sum_{k=1,2,3} \alpha_k p_t(T_k), \nonumber
\end{equation}
where $\alpha_k\in\mathbb{R}$ and $T_k\in\mathscr{T}$, for $k=1,2,3$.

In particular, the vertical bars stand for the nominal $\alpha_k$ of the bond related to any maturity $T_k$. The value of any $\alpha_k$ is assessed by the scale on the right side vertical axis of the figure.

On the other hand, the line shows the behaviour of the functional, 
\begin{equation} \label{eq:blue_line_bond}
\mathcal{F}(T)\triangleq\mathbb{E}\bigg \lbrace \int_I\bigg\Vert \sum_{k=1,2,3}\alpha_k\sigma_t(T_k) -\alpha(T)\sigma_t(T) \bigg\Vert_{H}^2 (1-t)dt\bigg\rbrace,  \hspace{1cm} \text{for $T\in\mathscr{T}$}. 
\end{equation}
that is derived form (\ref{eq:optimization_T}) and whose values are assessed by the scale reported on the left side vertical axis of the graph.

We considered a correlated two-additive-factor Gaussian model governing the evolution of the short-rate process. More precisely, let  $(W_1,W_2)$ be a two-dimensional Wiener process with some instantaneous correlation $\varrho \in (-1,1)$. 
The risk-neutral dynamics of the instantaneous-short-rate $r\triangleq \lbrace r_t: t\in I \rbrace$ is given by 
\begin{equation}
r_t \triangleq \chi_{1,t} + \chi_{2,t} + \varphi(t), \hspace{1cm} \text{for any $t\in I$}, \nonumber
\end{equation}
where for $i=1,2$ the process $\chi_i\triangleq \lbrace \chi_{i,t}:t\in I \rbrace$ is given by the following Vasicek-type model,
\begin{equation} \label{eq:g++2_factors}
\chi_{i,t} = \chi_{i,0} - \int_0^t a_i \chi_{i,t} dt +\int_0^t \sigma_i dW_{i,t}, \hspace{1cm} \text{for any $t\in I$,}
\end{equation}
for some given positive parameters $a_i$ and $\sigma_i$ jointly with the initial condition $\chi_{i,0}\in \mathbb{R}^+$, and where $t\in I\mapsto\varphi(t)$ is some deterministic function.
We refer to Section 4.2. in \cite{brigo} for all the details. 
 
This model is analytically manageable enough to write down an explicit formula for the discounted price curve in terms of the short-rate factors (\ref{eq:g++2_factors}), and hence to determine an analytical expression for the process $\sigma$ by considering standard calculus techniques. 
For any $T\in\mathscr{T}$, the value $\mathcal{F}(T)$ has been obtained as a combination of the Monte Carlo simulation of the factors (\ref{eq:g++2_factors}) jointly with the discretization of the integral related to the time variable $t\in I$.
All the details of the simulation are collected in the caption of Figure \ref{fig:bond_replication}.

%\begin{figure}[h]\label{fig:bond}
%\caption{...}
%  \centering
%    \includegraphics[width=0.6\textwidth]{g_plusplus1.png}
%\end{figure}

\

\section{Life insurance model points portfolio.}  \label{sec:life_insurance}
European insurance companies are required to assess the value of their portfolios as well as to carry on the sensitivity analysis aimed at demonstrating the compliance of their models, by considering the cash flow projections on a policy-by-policy approach. Besides, they are allowed to compute these projections by replacing any homogeneous group of policies with some suitable representative contracts, usually known as the related model points. This approach is aimed at speeding up this process, that is usually carried out on a daily basis, since the complexity of the entire portfolio may lead to long computational times. This procedure is permitted under suitable conditions in such a way that the inherent risk structure of the original portfolio is not misrepresented. We refer to \cite{eiopa} for any detail.

Through this section, we assess the problem of determining an optimal model points portfolio related to some fixed policies portfolio.

\

\noindent{\emph{Framework and notations.}}
Life insurance contracts usually provide either a stream of cash flows during the lifetime of the policyholder or a unique lump sum benefit that is paid upon his death under certain conditions. %According to the clauses of the contract, other events like the appearance of a terminal illness may also lead to an early payment. 
The contract is in force when the policyholder pays a specific premium, that depending upon the case may be either regular or as one initial lump sum. 

We do not consider the specific characteristics of the contract and we assume to deal with some idealized life insurance policy that is unaffected by credit risk, i.e. the insurance company always guarantees the entire benefit that is provided for in the contract. Further, we do not analyse the revenues received by the insurance company and thus we do not take into account the premium of the contract such as any further payment which is the responsibility of the client.

Let $\mathscr{X}$ be some set in which any element $x\in \mathscr{X}$ represents a contract. To fix the ideas, any $x\in \mathscr{X}$ may collect typical characteristics such as age and gender of the policy owner, cancellation option, etc. In other terms, any element $x\in\mathscr{X}$ identifies a class of policies that is labelled by considering certain suitable characteristics.

We suppose that there exists a market value of the contract relative to any $x\in \mathscr{X}$.

\
%
%Let $\mathscr{X}$ be a set and fix $\mathscr{Y}\subseteq \mathscr{X}$. Thus, assume any element $x\in \mathscr{X}$ to represent a reference individual. To fix the ideas, any $x\in \mathscr{X}$ may collect typical characteristics such as age, gender, cancellation option, etc. 
%Given this formulation, we may consider any portfolio composed by the same life insurance contract as index by the $\mathscr{X}$.
%We assess the problem to replace a fixed portfolio composed by homogeneous life insurance policies indexed by $\mathscr{X}$,  by considering a smaller portfolio of representative contracts indexed by $\mathscr{Y}$. Any $y\in \mathscr{Y}$ is said to be a \emph{model point}. 

\noindent{\emph{Model points portfolios.}}
Let $U$ be a UMD Banach space that is represented by some space of real valued functions defined on $\mathscr{X}$.  We write $U^\ast$ to denote the topological dual of $U$. The duality pairing between $U$ and $U^\ast$ is denoted by $\langle\cdot,\cdot\rangle_U$. 
Further, we assume the evaluation functional $\delta_x$ to be continuous and bounded on $U$, for any $x\in \mathscr{X}$.

We write $\mathscr{X}^\ast$ to denote the subsets of $U^\ast$ defined by setting 
\begin{equation}\label{eq:X^star}
\mathscr{X}^\ast\triangleq \overline{\text{span}\lbrace \delta_x :  x\in \mathscr{X}\rbrace},
\end{equation}
where the closure in (\ref{eq:X^star}) is to be understood with respect to the topology of $U^\ast$.

Here and in the sequel, we fix $\mathscr{Y}^\ast\subseteq \mathscr{X}^\ast$.

\begin{definition}
Let $v\in \mathscr{X}^\ast$ and consider a $U$-valued process $z\triangleq \lbrace z_t:t\in I \rbrace$. For any $w\in \mathscr{Y}^\ast$, the process $V_t(w)\triangleq \lbrace V_t(w):t\in I\rbrace$ defined by setting
\begin{equation}\label{eq:financial_exposure_model_point}
V_t(w) = \langle v, z_t \rangle_U- \langle w, z_t\rangle_U, \hspace{1cm} \text{for any $t\in I$,}
\end{equation}
is said to be the \emph{discrepancy process} between $v$ and $w$ relative to $z$.
\end{definition}
If not otherwise specified, where we fix a $U$-valued process $z\triangleq \lbrace z_t:t\in I \rbrace$, for any $v\in \mathscr{X}^\ast$ and $w\in \mathscr{Y}^\ast$, we always write $V(w)$ to denote the discrepancy process between $v$ and $w$ relative to $z$. 

For our application, we regard any process $z=\lbrace z_t : t\in I \rbrace$ as a dynamics for the discounted value of some specific life insurance contract. More precisely, we understand $z_t(x)$ as the risk-neutral discounted value of the contract at time $t\in I$, when it is referred to $x\in \mathscr{X}$. 
Besides, similarly to the framework of Section \ref{sec: Interest Rates}, it is worth to be noted that 
\begin{equation}
z_t(x)=\langle \delta_x,z_t \rangle_U, \hspace{1cm} \text{for any $t\in I$ and $x\in \mathscr{X}$.} \nonumber
\end{equation} 
Hence, we may regard $\delta_x$ as a portfolio composed by one policy related to $x\in \mathscr{X}$ and any $v\in \mathscr{X}^\ast$ as a portfolio composed by different policies.

On the other hand, we regard any $w\in \mathscr{Y}^\ast$ as a model points portfolio.

\begin{lemma}
Let $z\triangleq \lbrace z_t :t\in I \rbrace$ be a $U$-valued process such that $\sup_{t\in I}\Vert z_t\Vert^2_{L^2(\Omega;U)}< \infty$ and fix $v\in\mathscr{X}^\ast$.  Then, $V_t(w)\in L^2(\Omega)$, for any $t\in I$ and $w\in \mathscr{Y}^\ast$, with
\begin{equation}
\sup_{t\in I} \Vert V_t(w) \Vert^2_{L^2(\Omega)} < \infty. \nonumber
\end{equation}
\begin{proof}
Fix $w\in \mathscr{Y}^\ast$ and notice that 
\begin{equation}
\sup_{t\in I}\Vert V_t(w) \Vert_{L^2(\Omega)}^2 = \sup_{t\in I}\Vert \langle v-w,z_t \rangle_U \Vert_{L^2(\Omega)}^2  \leq \Vert v-w \Vert_{U^\ast}^2 \sup_{t\in I} \Vert z_t \Vert_{L^2(\Omega;U)}^2.  \nonumber
\end{equation}
\end{proof}
\end{lemma}

In view of our applications, we recast Definition \ref{de:risk_functional_F} and Definition \ref{de:optimality} as follows. 
\begin{definition} 
Let $z\triangleq \lbrace z_t : t\in I \rbrace$ be an $U$-valued process such that 
\begin{equation}\label{eq:z_sup_t_norm}
\sup_{t\in I}\Vert z_t\Vert^2_{L^2(\Omega;U)}< \infty.
\end{equation}
and fix $v\in\mathscr{X}^\ast$.
We refer to the functional $\mathcal{V}:\mathscr{Y}^\ast\rightarrow \mathbb{R}$ defined by 
\begin{equation}\label{eq:optimization_problem_mean_model_points}
\mathcal{V}(w)\triangleq \int_I \mathbb{E}\lbrace\vert V_t(w) - \mathbb{E}V_t(w) \vert^2\rbrace dt, \hspace{1cm} \text{for any $w\in \mathscr{Y}^\ast$,}
\end{equation}
as the \emph{model points risk functional} relative to $z$ induced by $v$ over the set $\mathscr{Y}^\ast$. 
Moreover, we call \emph{optimal model points portfolio} any $\mathcal{V}$-optimal element $w^\ast\in \mathscr{Y}^\ast$.
\end{definition}

If not otherwise specified, where we fix a process $z\triangleq \lbrace z_t: t \in I \rbrace$ satisfying the condition (\ref{eq:z_sup_t_norm}) and a portfolio $v\in \mathscr{X}^\ast$ is given, we always write $\mathcal{V}$ to denote the model points risk functional relative to $z$ induced by $v$ over $\mathscr{Y}^\ast$.

\

\noindent{\emph{Model points risk functional representation.}}
Let $E$ and $H$ be the spaces as considered in Section \ref{sec: Interest Rates} and thus assume the process $p=\lbrace p_t : t\in I \rbrace$ defined by the identity (\ref{eq:discounted_price_dynamics}) to model the risk-neutral dynamics of the discounted price curve.
Since the value of any policy is required to be estimated by considering the probability weighted average of the related future cash-flows, it is reasonable to write 
\begin{equation}
z_t=\zeta(t,p_t) \hspace{1cm} \text{for any $t\in I$,} \nonumber
\end{equation} 
for some function $\zeta:I\times E\rightarrow U$.

We discuss this approach later on.

\begin{lemma}\label{le:z_L2}
Fix a function $\zeta:I\times E \rightarrow U$ and denote by $p$ the process (\ref{eq:discounted_price_dynamics}). If $\zeta$ is of class $\mathscr{C}_b^{1,2}$, then $\zeta(t,p_t)\in L^2(\Omega;U)$ is well defined, for any $t\in I$, with
\begin{equation}
\sup_{t\in I} \Vert \zeta(t,p_t) \Vert^2_{L^2(\Omega;U)} < \infty. \nonumber
\end{equation} 
\begin{proof}
First, notice that since $\zeta(\cdot,0)$ is assumed to be continuous on $I$, we have that
\begin{equation}
\Vert \zeta(\cdot,0) \Vert_\infty\triangleq \sup_{t\in I} \Vert \zeta(t,0)\Vert_U < \infty \nonumber.
\end{equation}
Besides, since $\zeta$ is assumed to be of class $\mathscr{C}_b^{1,2}$, for any $t\in I$, the following inequalities hold true a.s., 
\begin{equation}
\Vert \zeta(t,p_t)\Vert_U^2 \leq  \Vert \nabla_2 \zeta\Vert^2_\infty \Vert p_t \Vert^2_E + \vert \zeta(t,0)\vert^2< \infty, \nonumber
\end{equation}
and hence
\begin{equation}
\sup_{t\in I}\Vert \zeta(t,p_t)\Vert_U^2 \leq  \Vert \nabla_2 \zeta\Vert^2_\infty \sup_{t\in I}\Vert p_t \Vert^2_E + \Vert \zeta(\cdot,0) \Vert^2_\infty < \infty, \nonumber
\end{equation}
since $\sup_{t\in I}\Vert p_t \Vert^2_E < \infty$ due to Lemma \ref{le:xi_well_defined}, applied to the process $p$.  
\end{proof}
\end{lemma}

\begin{proposition}\label{pr:model_point_characterization}
Fix $v\in \mathscr{X}^\ast$ and let $\zeta: I\times E\rightarrow U$ be a function of class $\mathscr{C}_b^{1,2}$. If $\zeta$ is a BS-function relative to $p$, where the process $p$ is given by (\ref{eq:discounted_price_dynamics}), and we set $z_t=\zeta(t,p_t)$, for any $t\in I,$ then
\begin{equation}\label{eq:model_point_characterization}
\mathcal{V}(w) = \mathbb{E}\bigg\lbrace\int_I\Vert \langle v - w,\nabla_2 \zeta(t,p_t)\sigma_t \rangle_U \Vert_{H}^2 (1-t)dt\bigg\rbrace, \hspace{1cm} \text{for any $w\in \mathscr{Y}^\ast$.}
\end{equation}
\end{proposition}

It is worth to be noted that, since the process $\sigma$ takes values in  $\gamma(H,E)$, Lemma \ref{le:ideal_property} gives that the process $\nabla_2 \zeta (t,p_t)\sigma_t$, for $t\in I$, takes values in $\gamma(H,U)$. Thus, in the representation (\ref{eq:model_point_characterization}) we regard $\langle\cdot,\cdot \rangle_U$ as the $H$-valued pairing between $\gamma(H,U)$ and $U^\ast$.

\begin{proof}[Proof of Proposition \ref{pr:model_point_characterization}]
First, notice that $\sup_{t\in I}\Vert z_t\Vert^2_{L^2(\Omega;U)}< \infty$ thanks to Lemma \ref{le:z_L2}, since the function $\zeta: I\times E\rightarrow U$ is assumed to be of class $\mathscr{C}_b^{1,2}$. 

On the other hand, notice that since $\zeta$ is assumed to be a BS-function relative to $p$, and the process $p$ is given by (\ref{eq:discounted_price_dynamics}), Lemma \ref{le:BL-function} gives that a.s.
\begin{equation}\label{eq:price_surface}
\zeta(t,p_t) = \zeta(0,p_0) + \int_0^t \nabla_2 \zeta(s,p_s)\sigma_s dW_s, \hspace{.5cm} \text{for any $t\in I$.}  
\end{equation}

Let $\Phi$ to be a P-set relative to $z$ defined in such a way that for any $\phi\in \Phi$ there exists $w\in\mathscr{Y}^\ast$ with $\phi_t= w$, a.s., for any $t\in I$. 
Then, the result follows from the statement (ii) in Theorem \ref{pr:optimal_strategy} by setting $\xi_t=\zeta(t,p_t)$ and $f(t,u)= \langle v, u \rangle_U$, for any $u\in U$. In this respect, notice that $b_t=0$ a.s., for any $t\in I$, and $f$ appears to be a BS-function relative to $z$. 
\end{proof}

\ 

\noindent{\emph{Whole life insurance.}}
Throughout, we discuss the problem of the model points selection when dealing with a portfolio of whole life insurance policies, i.e. a contract that provides for a one unit benefit on the death of the policy owner. 

In view of our application, we assume $\mathscr{X}$ to coincide with some closed interval of $\mathbb{R}^+$ and we regard any $x\in\mathscr{X}$ as the age of the policy owner at time $t=0$. 
Moreover, here and in the sequel of this part we define $U\triangleq \mathbb{W}^{1,2}(\mathscr{X})$ to be the space of absolutely continuous functions $u:\mathscr{X}\rightarrow \mathbb{R}$ such that $\Vert u' \Vert_{L^2(\mathscr{X})}<\infty$, where we denoted by $u'$ the weak derivative of $u$.
In this respect, recall that $U$ is a Hilbert space, when endowed with the norm
\begin{equation}
\Vert u \Vert_{U} \triangleq \lbrace \Vert u \Vert^2_{L^2(\mathscr{X})} + \Vert u'\Vert^2_{L^2(\mathscr{X})} \rbrace^{1/2}, \hspace{1cm} \text{for any $u\in U$,} \nonumber
\end{equation}
such that the evaluation functional $\delta_x(u)\triangleq u(x)$, for $u\in U$, is bounded for any $x\in\mathscr{X}$. 

We write $\mu(s,x+s)$ to denote the force of mortality relative to $x\in \mathscr{X}$ at any time $s\geq 0$. Moreover, we do not consider those deaths that occur during the first year by setting
\begin{equation}\label{eq:mortality_condition}
\mu(s,x+s)=0, \hspace{1cm}\text{for any $x\in \mathscr{X}$ and for $s\in I$,} 
\end{equation} 
and hence defining the survivor index as
\begin{equation}\label{eq:survivir_index}
S(x,T)\triangleq \exp \bigg\lbrace - \int_1^T \mu(s,x+s)ds \bigg\rbrace, \hspace{1cm} \text{for any $x\in \mathscr{X}$ and $T\in \mathscr{T}$.}
\end{equation}
We regard $S(x,T)$ as the proportion of individuals that are $x$-aged at time $t=0$ and that survive to age $x+T$. 

It is worth to be highlighted that the condition (\ref{eq:mortality_condition}) is convenient for our application, since it implies that the policy portfolio does not change during the time interval $I$ due to the death of the policy owners. Such an assumption is acceptable, since the events occurring within the first year only cause a minimal impact on the performance of the overall portfolio.

The discounted value of a whole life insurance policy relative to $x\in\mathscr{X}$ at time $t\in I$ may be written as
\begin{equation}\label{eq:while-life-insurance}
z_t(x)= \int_\mathscr{T}  S(x,T) \mu(T,x+T)
p_t(T) dT. 
\end{equation}
The following Lemma tell us that, under mild conditions, the identity (\ref{eq:while-life-insurance}) provides a well defined $U$-valued process.
\begin{lemma}
Let $w:\mathscr{T}\rightarrow \mathbb{R}^+$ be an increasing function satisfying (\ref{eq:condition_w_H}) and set $E=\mathcal{H}_w$ according to Definition \ref{de:H_w}. 
Further, assume the function $x\in \mathscr{X}\mapsto \mu(s,x+s)$ to be continuously differentiable, for any $s\in \mathscr{T}$. 
If $w\geq 1$ everywhere on $\mathscr{T}$ and the following condition holds
\begin{equation}\label{eq:S_bar_condition}
\sup_{T\in \mathscr{T}} w(T)^{-1} \int_{\mathscr{X}} \vert \partial_x S(x,T) \vert^2 dx < \infty, 
\end{equation}
then the process $z\triangleq \lbrace z_t : t\in I \rbrace$ given by the identity (\ref{eq:while-life-insurance}) is $U$-valued, with
\begin{equation}\label{eq:sup_z_norm_lemma}
\sup_{t\in I}\Vert z_t\Vert^2_{L^2(\Omega;U)}< \infty. 
\end{equation}
\end{lemma}
\begin{proof}
First, notice that since $p_t(T)\rightarrow 0$ a.s., as $T\rightarrow +\infty$, for any $t\in I$, and 
\begin{equation}
S(x,T)\mu(T,x+T) = - \partial_T S(x,T), \hspace{1cm} \text{for any $x\in \mathscr{X}$ and $T\in \mathscr{T}$}, \nonumber
\end{equation}
the identity (\ref{eq:while-life-insurance}) may be recast by invoking integration by part arguments as follows,  
\begin{equation}
z_t(x)= -S(x,1)p_t(1) + \int_{\mathscr{T}} S(x,T)p'_t(T) dT, \hspace{1cm} \text{for any $t\in I$ and $x\in \mathscr{X}$}. \nonumber
\end{equation}
As a direct consequence, for any fixed $t\in I$, 
\begin{eqnarray} \label{eq:z_L2_norm}
\int_\mathscr{X} z_t(x)^2 dx &\leq & \int_{\mathscr{X}} S(x,1)^2p_t(1)^2 dx  + \int_\mathscr{X}\int_{\mathscr{T}} S(x,T)^2p_t'(T)^2 dTdx \nonumber  \\
&\stackrel{\text{(i)}}{\leq} &  m(\mathscr{X})\bigg \lbrace p_t(1)^2 + \int_{\mathscr{T}} p_t'(T)^2 dT \bigg\rbrace,
\end{eqnarray}
where the inequality (i) holds true since $S(x,T)\leq 1$, for any $x\in \mathscr{X}$ and $T\in \mathscr{T}$, and the function $x\in\mathscr{X}\mapsto \mu(s,x+s)$ in continuous, for any $s\in \mathscr{T}$. Moreover, we denoted by $m(\mathscr{X})$ the Lebesgue measure of the interval $\mathscr{X}$. 

Notice that, Lemma \ref{le:xi_well_defined} applied to the process $p$ gives
\begin{equation}\label{eq:p_L2_z_proof}
\sup_{t\in I} \mathbb{E}\Vert p_t \Vert^2_E = \sup_{t\in I} \mathbb{E}\bigg \lbrace p_t(1)^2 + \int_\mathscr{T} p'_t(T)^2w(T)dT \bigg \rbrace < \infty. 
\end{equation}
%As a consequence, the Cauchy-Schwarz inequality jointly with Lemma \ref{le:condition_w_H} give
%\begin{equation} \label{eq:p1}
%\sup_{t\in I} \mathbb{E} \lbrace p_t(1)^2 \rbrace = \sup_{t\in I} \mathbb{E}\vert \langle \delta_1,p_t\rangle_{E}\vert^2 \leq \Vert \delta_1 \Vert^2_{E^\ast} \sup_{t\in I} \mathbb{E}\Vert p_t \Vert^2_E < \infty, 
%\end{equation}
%and H\"older's inequality jointly with the condition $w\geq 1$, everywhere on $\mathscr{T}$, lead to 
%\begin{equation}\label{eq:holder_p_w}
%\sup_{t\in I} \mathbb{E} \bigg\lbrace \int_{\mathscr{T}}p_t'(T)^2 dT \bigg\rbrace \leq \sup_{t\in I} \mathbb{E} \bigg\lbrace \int_{\mathscr{T}} p_t'(T)^2 w(T)dT \bigg\rbrace \sup_{T\in \mathscr{T}} w(T)^{-1} <\infty. 
%\end{equation}
Thus, inequality (\ref{eq:p_L2_z_proof}) combined with (\ref{eq:z_L2_norm}) gives
\begin{equation}
\sup_{t\in I} \mathbb{E} \Vert z_t \Vert^2_{L^2(\mathscr{X})}= \sup_{t\in I} \mathbb{E}\bigg\lbrace \int_\mathscr{X} z_t(x)^2 dx \bigg \rbrace < \infty. \nonumber
\end{equation}

On the other hand, since the function $x\in \mathscr{X}\mapsto \mu(s,x+s)$ is assumed to be continuously differentiable,
\begin{equation}
T\in\mathscr{T}\mapsto\int_{\mathscr{X}} \vert \partial_x S(x,T) \vert^2 dx, \nonumber
\end{equation} 
is well defined everywhere on $\mathscr{T}$. Further,
\begin{equation} \label{eq:z_prime_parts}
z'_t(x)= -\partial_x S(x,1)p_t(1) + \int_{\mathscr{T}} \partial_x S(x,T)p'_t(T) dT, \hspace{1cm} \text{for any $t\in I$ and $x\in \mathscr{X}$}, 
\end{equation}
and hence, for any fixed $t\in I$,
\begin{equation}\label{eq:z_prime_bound_1}
\int_\mathscr{X} z'_t(x)^2 dx \leq p_t(1)^2 \int_{\mathscr{X}} \vert\partial_x S(x,1)\vert^2 dx  + \int_\mathscr{T}\bigg\lbrace\int_{\mathscr{X}} \vert\partial_x S(x,T)\vert^2 dx \bigg\rbrace p_t'(T)^2 dT. 
\end{equation}
Then, since 
\begin{multline} 
\sup_{t\in I} \mathbb{E} \bigg\lbrace  \int_{\mathscr{T}} \bigg\lbrace\int_{\mathscr{X}} \vert\partial_x S(x,T)\vert^2 dx \bigg\rbrace p_t'(T)^2 dT \bigg\rbrace\leq \\
\sup_{t\in I} \mathbb{E} \bigg\lbrace  \int_{\mathscr{T}} p_t'(T)^2  w(T) dT \bigg\rbrace \sup_{T\in\mathscr{T}} w(T)^{-1} \int_{\mathscr{X}} \vert \partial_x S(x,T) \vert^2 dx < \infty, \nonumber
\end{multline}
thanks to H\"older's inequality again combined with the assumption (\ref{eq:S_bar_condition}) and the condition (\ref{eq:p_L2_z_proof}), according to (\ref{eq:p1})
 we get that (\ref{eq:z_prime_bound_1}) leads to
\begin{equation}
\sup_{t\in I} \mathbb{E}  \Vert z'_t \Vert^2_{L^2(\mathscr{X})}= \sup_{t\in I} \mathbb{E}\bigg\lbrace \int_\mathscr{X} z'_t(x)^2 dx \bigg \rbrace < \infty. \nonumber
\end{equation}
\end{proof}

For some $K\in \mathbb{N}$, we fix $x_1,...,x_K\in \mathscr{X}$ and $\alpha_1,...,\alpha_K \in \mathbb{R}^+$ and we consider the portfolio $v\in \mathscr{X}^\ast$ defined by setting
\begin{equation}\label{eq:life_insurance_fixed_portfolio}
v\triangleq \sum_{k=1}^K \alpha_k \delta_{x_k}.
\end{equation}
%Hence, jointly with representation (\ref{eq:while-life-insurance}), the risk-neutral discounted value of the portfolio $v$ at time $t\in I$ is given by
%\begin{equation}
%\langle  v, z_t \rangle_U 
%= \sum_{k=1}^K \alpha_k \int_1^{+\infty}  S(x_k,T) \mu(x_k+T)p_t(T)
% dT.
%\end{equation}
On the other hand, assume that any $w\in \mathscr{Y}^\ast$ admits the following representation 
\begin{equation} \label{eq:model_point_representaion_y}
w = \alpha(x)\delta_x
\end{equation}
for some $x\in \mathscr{X}$, where we set, 
\begin{equation} \label{eq:model_point_time_0}
\alpha(x)\triangleq z_0(x)^{-1} \langle v, z_0 \rangle_U. 
\end{equation}
As a direct consequence, we may write $\mathscr{Y}^\ast=\mathscr{X}$, by identifying any $w\in \mathscr{Y}^\ast$ with the element $x\in \mathscr{X}$ that satisfies the representation (\ref{eq:model_point_representaion_y}).

It is worth to be noted that the condition (\ref{eq:model_point_time_0}) is reasonable, since it guarantees that 
\begin{equation}
\langle v,z_0 \rangle_U = \langle w,z_0 \rangle_U \nonumber
\end{equation}
and hence that any model points portfolio $w\in \mathscr{Y}^\ast$ admits the same discounted value as $v\in\mathscr{X}^\ast$, at time $t=0$.

%
%Hence, we say that the target client $y\in \mathscr{Y}$ is $V$-optimal if the following condition is satisfied 
%\begin{equation}\label{eq:model_point_characterization}
%y^\ast =\operatorname*{argmin}_{y\in \mathscr{Y}} \mathbb{E}\bigg\lbrace\int_I\bigg\Vert \sum_{k=1}^K \alpha_k \nabla_2^2 \zeta(t,p_t)\sigma_t(x_k)- \alpha(y)\nabla_2^2 \zeta(t,p_t)\sigma_t(y) \bigg\Vert_{H}^2 (1-t)dt\bigg\rbrace.
%\end{equation}
%

\begin{proposition}\label{pr:blue_line_life}
Let $v\in \mathscr{X}^\ast$ as defined by the identity (\ref{eq:life_insurance_fixed_portfolio}) and $\mathscr{Y}^\ast$ such that the representation (\ref{eq:model_point_representaion_y}) holds true for any $w\in\mathscr{Y}^\ast$. Further, let $z=\lbrace z_t : t\in I \rbrace$ be the process defined by the identity (\ref{eq:while-life-insurance}). Then,
%\begin{multline} \label{eq:blue_line_life}
%\mathcal{V}(y)=\mathbb{E}\bigg\lbrace\int_I\bigg\Vert \sum_{k=1}^K \alpha_k \nabla_2 \zeta(t,p_t)\sigma_t(x_k)- \alpha(y)\nabla_2 \zeta(t,p_t)\sigma_t(y) \bigg\Vert_{H}^2 (1-t)dt\bigg\rbrace, \\ \text{for any $y\in \mathscr{Y}$.} 
%\end{multline}
%\end{proposition}
\begin{multline} \label{eq:blue_line_life}
\mathcal{V}(x)=\mathbb{E}\bigg\lbrace\int_I\bigg\Vert 
\int_\mathscr{T} \bigg(
\sum_{k=1}^K \alpha_k \kappa(x_k,T)
  -  \alpha(x) \kappa(x,T)\bigg)\sigma_t(T)
 dT \bigg\Vert_{H}^2 (1-t)dt\bigg\rbrace, \\\hspace{1cm}\text{for any $x\in \mathscr{X}$,} 
\end{multline}
where we set
\begin{equation}
\kappa(x,T)\triangleq S(x,T)\mu(T,x+T), \hspace{1cm} \text{for any $x\in\mathscr{X}$ and $T\in \mathscr{T}$.} \nonumber
\end{equation}
\end{proposition}
\begin{proof}
First of all, we introduce the functional $Z\in \mathcal{L}(E,U)$ defined by setting 
\begin{equation}
Z(q)= \int_\mathscr{T} \kappa(\cdot, T) q(T) dT, \hspace{1cm} \text{for any $q\in E$.} \nonumber
\end{equation} 
Hence, when considering the function $\zeta:I\times E\rightarrow U$ obtained by the identity 
\begin{equation}\label{eq:Z_zeta_representation}
\zeta(t,q)\triangleq Z(q), \hspace{1cm} \text{for any $t\in I$ and $q\in E$,} 
\end{equation}
we get that $\zeta$ is of class $\mathscr{C}_b^{1,2}$ and it turns out to be a BS-function relative to $p$.  

Moreover, note that the process $z=\lbrace z_t:t\in I \rbrace$ defined by the identity (\ref{eq:while-life-insurance}) is recovered by setting $z_t=Z(p_t)$, for any $t\in I$.
In this respect, it is worth to be highlighted that the identification (\ref{eq:Z_zeta_representation}) is allowed since the survivor index (\ref{eq:survivir_index}) does not depend on $t\in I$, that is the case when imposing the condition (\ref{eq:mortality_condition}).

Thus, when writing $\nabla Z$ to denote the operator $\nabla_2\zeta$,  Proposition \ref{pr:model_point_characterization} applies and gives that
\begin{multline}\label{eq:blue_line_proof}
\mathcal{V}(x)=\mathbb{E}\bigg\lbrace\int_I\bigg\Vert \sum_{k=1}^K \alpha_k \nabla Z(p_t)\sigma_t(x_k)- \alpha(x)\nabla Z(p_t)\sigma_t(x) \bigg\Vert_{H}^2 (1-t)dt\bigg\rbrace, \\ \text{for any $x\in \mathscr{X}$.} 
\end{multline}
where for notation simplicity, in the identity (\ref{eq:blue_line_proof}) we set 
\begin{equation}
\nabla Z(p_t)\sigma_t(x)\triangleq \langle \delta_x,\nabla Z(
p_t)\sigma_t \rangle_U, \hspace{1cm} \text{for any $t\in I$ and $x\in \mathscr{X}$.} \nonumber
\end{equation}

On the other hand, by a direct computation, one has that a.s.
\begin{equation} \label{eq:model_point_deirivative}
\nabla Z(p_t)\sigma_t(x) = \int_\mathscr{T}  \kappa(x,T)\sigma_t(T)dT, \hspace{1cm} \text{for any $x\in\mathscr{X}$} \nonumber  
\end{equation}
and hence, jointly with the identity (\ref{eq:blue_line_proof}) we obtain the representation (\ref{eq:blue_line_life}).
\end{proof}

\begin{figure}
    \centering
        \includegraphics[width=0.58\textwidth]{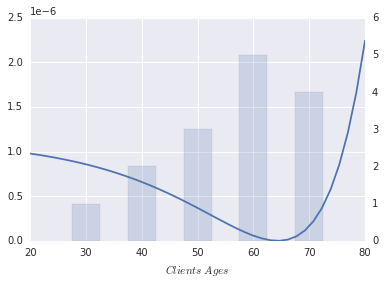}
    \caption{Visual representation of the functional (\ref{eq:model_point_characterization}), where $K=5$, with $x_1=30$, $x_2=40$,...,$x_K=70$. 
In particular, the values $\alpha_k$, for any $k=1,...,K$, are represented by the vertical bar and they are assessed by the right side vertical axis. The value $\mathcal{V}(x)$, for $20\leq x \leq 80$, is represented by the blue line and it is assessed by the left side vertical axis.
The process (\ref{eq:g++2_factors}) has been considered to model the short rate dynamics with same parameters as described in the caption of Figure \ref{fig:bond_replication}. For simplicity, the mortality force (\ref{eq:mortality_force}) has been computed by setting $a(s)=0.0003$ and $b(s)=0.06$, for any $s\geq 1$.} \label{fig:whole_life_insurance}
\end{figure}

\noindent{\emph{Numerical example}.}
Figure \ref{fig:whole_life_insurance} provides a visual representation of the functional (\ref{eq:blue_line_life}). In particular, each bar represents the amount $\alpha_k$ associated to the age $x_k$, for $k=1,...,K$, and it is assessed by the right side vertical axis of the figure. On the other hand, the functional $\mathcal{V}(x)$, varying $x\in \mathscr{X}$, is represented by the line and its value is assessed by the scale on the left side vertical axis of the figure.

The process $p$ has been simulated by considering the same model  governing the evolution of the instantaneous-short-rate as in Section \ref{sec: Interest Rates}. 
The survivor index (\ref{eq:survivir_index}) has been derived by considering a Gompertz-type law modelling the force of mortality \cite{gom}, defined by setting
\begin{equation}\label{eq:mortality_force}
\mu(s,x+s)= a(s) \exp{\lbrace(x+s) b(s)}\rbrace, \hspace{1cm} \text{for any $x\in \mathscr{X}$ and $s\in\mathscr{T}$.} 
\end{equation} 
where $a(s)$ and $b(s)$, varying $s\in\mathscr{T}$, are positive functions.

All the details of the simulation are collected in the caption of Figure \ref{fig:whole_life_insurance}.

%
%A \emph{term insurance} is an insurance policy that pays a lump sum benefit on the death of the policyholder, in return of required premiums that are paid up to a maximum age, provided that the death occurs before a term specified in the contract. 
%The discounted value of a term insurance at time $t\in I$ when owned by an $x$-aged individual at time $t=0$ and with term $\tau \geq 1$ is given by
%\begin{equation}
%z_t(x,\tau)\triangleq -\int_1^{\tau} p_t(T) \partial_T S(x,T)
% dT. \nonumber
%\end{equation}
%In this case, one can set $\mathscr{X}=\mathbb{R}^+\times \mathbb{R}^+$ and assume $(x,\tau)\in \mathscr{X}$ to represent both the age of the representative individual at time $t=0$ and the term of the contract.
%
%A \emph{single premium deferred annuity} is an insurance contract that provides a regular cash flow starting at some future date, in return to a lump-sum payment by the owner.
%The discounted value of a single premium deferred annuity at time $t\in I$ when owned by an $x$-aged individual at time $t=0$ is given by
%\begin{equation}
%z_t(x)= \int_1^{+\infty}  p_t(T) S(x,T) dT.
%\end{equation}
%As in the case of the whole life insurance, one can set $\mathscr{X}=\mathbb{R}^+$ and assume $x\in \mathscr{X}$ to represent the age of the representative individual at time $t=0$.

\ 

\noindent{\textbf{Acknowledgements.}} The author would like to thank Jos\'e L. Fern\'andez for fruitful discussions and providing valuable suggestions which have greatly improved the exposition. The author also thanks Ana M. Ferreiro and Jos\'e A. Garc\'ia for a constructive help related to the computational issues of the paper.

\bibliographystyle{acm}
\bibliography{bibfile}
\end{document}